\documentclass[11pt]{amsart}
\usepackage{graphicx, amssymb}

\newtheorem{thm}{Theorem}[section]

\newtheorem{lem}[thm]{Lemma}
\newtheorem{prop}[thm]{Proposition}
\theoremstyle{definition}
\newtheorem{defn}[thm]{Definition}
\theoremstyle{remark}
\newtheorem{rem}[thm]{Remark}
\newtheorem{ex}[thm]{Example}
\newcommand{\abs}[1]{\left\vert#1\right\vert}

\newcommand{\Real}{\mathbb R}
\newcommand{\Natural}{\mathbb N}

\newcommand{\such}{\, | \, }
\newcommand{\prob}{\mathbb{P}}
\newcommand{\Exp}{\mathcal E}
\newcommand{\qprob}{\mathbb{Q}}
\newcommand{\expec}{\mathbb{E}}

\newcommand{\probtriple}{(\Omega, \mathcal{F}, \prob)}

\newcommand{\filtration}{\mathbf{F} = \pare{\mathcal{F}_t}_{t \in \Real_+}}
\newcommand{\F}{\mathcal{F}}
\newcommand{\G}{\mathcal{G}}

\newcommand{\ud}{\mathrm d}
\newcommand{\inner}[2]{\left \langle #1 , #2 \right \rangle}
\newcommand{\p}{\mathrm{p}}
\newcommand{\Leb}{\mathsf{Leb}}

\newcommand{\e}{\mathrm{e}}

\newcommand{\g}{\mathfrak{g}}
\newcommand{\rel}{\mathfrak{rel}}
\newcommand{\pare}[1]{\left(#1\right)}
\newcommand{\bra}[1]{\left[#1\right]}
\newcommand{\dbra}[1]{[\kern-0.15em[ #1 ]\kern-0.15em]}
\newcommand{\dbraco}[1]{[\kern-0.15em[ #1 [\kern-0.15em[}
\newcommand{\dbraoo}[1]{]\kern-0.15em] #1 [\kern-0.15em[}
\newcommand{\C}{\mathfrak{C}}
\newcommand{\K}{\mathfrak{K}}

\newcommand{\fone}{\mathbf{1}}
\newcommand{\bF}{\mathbf{F}}
\newcommand{\bG}{\mathbf{G}}
\newcommand{\indic}{\mathbb{I}}

\newcommand{\hx}{x \indic_{\{ |x| \leq 1 \}}}
\newcommand{\hbarx}{x \indic_{\{ |x| > 1 \}}}
\begin{document}

\title[Balance, growth and diversity of financial markets]{Balance, growth and diversity of financial markets}%
\author{Constantinos Kardaras}%
\address{Constantinos Kardaras, Mathematics and Statistics Department, Boston University, 111 Cummington Street, Boston, MA 02215, USA.}%
\email{kardaras@bu.edu}%
\thanks{JEL code classification: G14}%
%\subjclass[JEL]{}

\date{\today}%
%\dedicatory{}%
%\commby{}%
%----------------------------------------------------------------

\begin{abstract}
A financial market comprising of a
certain number of distinct companies is considered, and the
following statement is proved: either a specific agent will surely
beat the whole market unconditionally in the long run, or (and this
``or'' is not exclusive) all the capital of the market will
accumulate in one company. Thus, absence of any ``free unbounded
lunches relative to the total capital'' opportunities lead to the
most dramatic failure of diversity in the market: one company takes
over all other until the end of time. In order to prove this, we
introduce the notion of \textsl{perfectly balanced} markets, which
is an equilibrium state in which the relative capitalization of each
company is a martingale under the physical probability. Then, the
weaker notion of \textsl{balanced} markets is discussed where the
martingale property of the relative capitalizations holds only
\emph{approximately}, we show how these concepts relate to
growth-optimality and efficiency of the market, as well as how we
can infer a \emph{shadow} interest rate that is implied in the
economy in the absence of a bank.
\end{abstract}

\maketitle

% ----------------------------------------------------------------
\setcounter{section}{-1}

\section{Introduction}

\subsection{Discussion and results}

We consider a model of a financial market that consists of $d$
stocks of certain ``distinct'' companies. The distinction between
companies clings on their having different risk and/or growth
characteristics, and will find its mathematically precise definition
later on in the text.

In absence of clairvoyance, the \emph{total} capital of each company
is modeled as a stochastic process $S^i$, $i=1,\ldots, d$.
Randomness comes through a set $\Omega$ of possible \emph{outcomes}
--- for each $\omega \in \Omega$ we have different realizations of
$S^i(\omega)$. Financial agents decide to invest certain amounts of
their wealth to different stocks, and via their actions the value of
$S^i_t$ for each time $t \in \Real_+$ is determined.

Of major importance in our discussion will be the \emph{distribution
of market capital}, given by the \textsl{relative capitalization}
$\kappa^i := S^i / (S^1 + \ldots + S^d)$ of each company ($S^1 +
\ldots + S^d$ is the total market capital). In particular, the
limiting, i.e., long-run, capital distribution will be investigated.
For addressing this question, a probability $\prob$ is introduced
that weights the different outcomes of $\Omega$ (for all events in
some $\sigma$-algebra $\F$); $\prob$ reflects the \emph{average}
subjective feeling of the financial agents, but in this average
sense it is not subjective anymore: each agent's investment
decisions are fed back into the relative capitalization of the
companies, and thus affects the random choice of the outcome. Via
this mechanism, $\prob$ becomes a \emph{real-world} probability, and
can also be regarded as the subjective view of a
\emph{representative agent} in the market, whose decisions alone
reflect the cumulative decisions of all ``small'' agents.

The time-flow of information is modeled via a \textsl{filtration}
$\filtration$. Each $\sigma$-algebra $\F_t$ is supposed to include
all (economical, political, etc.) information gathered up to time
$t$ and is increasing in time: $\F_s \subseteq \F_t$ for $0 \leq s <
t < \infty$. A ``representative agent'' information structure cannot
be justified, since different agents might have \emph{very}
different ability or capability to access information. This
difficulty can be circumvented by choosing $\bF$ in a minimal way,
i.e., by assuming that it is exactly the information contained in
the company capitalizations --- it is reasonable to assume that
\emph{every} agent has at least access to this information. This
minimal information structure will turn out to be the most useful in
our discussion (exactly because of its minimality property).

\smallskip

An important question from a modeling point of view is: \emph{how
does one go about choosing $\prob$ in a reasonable way in order to
reflect the way financial agents act?} From the economical side, the
concept of \emph{efficiency} has been quite extensively discussed in
the literature. In his famous work \cite{Fama: eff}, Fama states
that \emph{a market in which prices `fully reflect' available
information is called `efficient'}. Thus, efficiency is a property
that the capitalization processes $S$ must have under the pair
$(\prob, \bF)$, but it is questionable whether it opens the door to
mathematically pin down what are the possible ``reasonable''
probabilities $\prob$.

In the field of Mathematical Finance it has been argued that a
minimal condition for efficiency is absence of ``free lunch''
possibilities for agents; for if a free lunch existed, a sudden
change in the capital distribution would occur to correct for it,
which would contradict the requirement that prices fully reflect
information. The notion of ``no free lunch'' found its mathematical
incarnation in the existence of a probability $\qprob$ that is
equivalent to $\prob$ (meaning that $\prob$ and $\qprob$ have the
same impossibility events) under which capitalization processes
\emph{suitably deflated} have some kind of martingale property under
$(\qprob, \bF)$. However, as already mentioned this is only a
\emph{minimal} condition for efficiency. Indeed, consider a
two-stock market in which deflated capitalization processes are
modeled by $S_t^1 = \exp (W^1_t)$ and $S_t^2 = \exp (100t + W^2_t)$
where $t \in [0, T]$ for some finite $T$, and $(W^1, \, W^2)$ is a
2-dimensional standard Brownian motion. An equivalent martingale
measure $\qprob$ as described above exists for this model.
Nevertheless, these being the only two investment opportunities in
the market, reasonable agents would opt for the second choice over
the first. Even if diversification was sought-after, significantly
more capital would be held in the second rather than the first
stock. This huge movement of capital would change the capitalization
dynamics
--- this market does not appear to be in equilibrium, it is not
balanced.

As mentioned previously, coupled with the choice of an equivalent
martingale measure comes the choice of a \textsl{deflator} in the
market. It is a usual practice to use the interest rate offered for
risk-free investments for discounting. Nevertheless, it is
questionable whether the interest-rate structure reflects the true
market growth; a better index has to be perceived
--- and what would be more reasonable to use than the \emph{total
market capital}? Directly considering the \emph{percentage} of the
total capitalization that each individual company occupies, its
performance in terms of the ``competing'' ones is assessed.

In the spirit of the above discussion, the idea of a
\textsl{perfectly balanced market} is formulated by requiring that
\emph{the relative capitalizations $\kappa^i$ are martingales under
$(\prob, \bF)$}: $\expec[ \kappa^i_t \such \F_s] = \kappa^i_s$, for
all $i=1, \ldots, d$ and $0 \leq s < t$. The last equality means
that the best prediction about the future value of the relative
capitalization of a company given today's information is exactly the
present value of the relative capitalization. One might ask why is
this martingale property plausible. Consider, for example, what
would happen if $\expec[ \kappa^i_t \such \F_s] < \kappa^i_s$ for
some company $i$. Since at all times the sum of all the relative
capitalizations should be unit, we have $\expec[ \kappa^j_t \such
\F_s]
> \kappa^j_s$ for another company $j$. These inequalities suggest that the overall feeling
of the market is that in the future (time $t$) the $i^\textrm{th}$
company will hold on average a smaller piece of the pie than it does
today (time $s$), with the converse holding for company $j$ --- in
other words, that company $i$ is presently overrated, while company
$j$ underrated. The reasonable thing to happen is a movement of
capital from company $i$ to company $j$, which would move
$\kappa_s^i$ downwards and $\kappa_s^j$ upwards, until finally
$\expec[ \kappa^i_t \such \F_s] = \kappa^i_s$ holds for all $i=1,
\ldots, d$.

Perfect balance, as an equilibrium state, can undergo much
criticism: there will certainly be times at which the market
``slides away'' from being perfectly balanced, but it would be
reasonable to assume that the market is quickly trying to readjust
itself to that state (as was explained in some sense in the previous
paragraph). A mathematically rigorous description of this concept
would require a formulation of an ``approximate martingale''
property for the relative capitalization vector $\kappa$. The
widely-accepted idea of assuming the existence of another
probability $\qprob$ that is equivalent to $\prob$, and such that
$\kappa$ is a martingale under $\qprob$ seems to be appropriate
(actually, this exact idea has been utilized in Yan \cite{Yan}, who
has shown its equivalence to a ``no-free-lunch'' property relative
to the total capitalization $\sum_{i=1}^d S^i$), as long as $\qprob$
and $\prob$ are ``close'' in some sense . This is \emph{not} the
road that will be taken here, and there are at least two good
reasons: firstly, some (necessarily) ad-hoc, as well as difficult to
justify in economic terms, definition of distance between $\prob$
and $\qprob$ would have to be given; secondly, existence of such a
$\qprob$ is \emph{not} an $\omega$-by-$\omega$ notion (as it looks
at \emph{all} possible outcomes instead), and after all what shall
be ultimately revealed is only one outcome. However, an
$\omega$-by-$\omega$ definition of plainly \textsl{balanced markets}
(based on a characterization of perfectly balanced markets given by
observable quantities of the model) comes to the rescue --- in some
sense to be made precise later, the market is balanced if the
process $\kappa$ is \emph{close} to being a martingale, but not
quite there. The notion of balanced markets will turn out to be
\emph{strictly} weaker than the requirement of existence of such
probability $\qprob$ as described above in this paragraph.

Having decomposed the state space $\Omega$ as $\Omega_b \cup
\Omega_u$, where $\Omega_b$ is the set of outcomes where the market
is balanced and its complement $\Omega_u$ is the set of outcomes
that it is unbalanced, an analysis of the behavior of the market on
each of the above two events is in order. It turns out that on
$\Omega_u$ a \emph{single} agent can beat the whole market for
arbitrary levels of wealth, an unacceptable situation since the
total capital of the market \emph{should} consist of the sum of the
wealths of its respective agents; on the unbalanced set this breaks
down, since one \emph{particular} agent will eventually have more
capital than the whole market. It then makes sense to focus on the
balanced-market outcomes $\Omega_b$. There, it turns out that there
always exists a limiting distribution of capital $\kappa_\infty$ in
the almost sure sense. If one further assumes that the market is
\textsl{segregated}, in the sense that companies are distinct in a
very weak sense, it turns out that all capital will concentrate in a
\emph{single} company. This is probably the most dramatic failure of
\textsl{market diversity} pioneered by Fernholz \cite{Fernholz:
SPT}. In this last monograph, as well as in Fernholz, Karatzas and
Kardaras \cite{FKK}, it is shown that certain diverse markets offer
opportunities for free lunches relative to the market. Taking up on
this, the present work shows that failure of diversity
\emph{inevitably} leads to free lunches relative to the market ---
at the opposite direction, non-existence of free lunches (relative
to the market) \emph{a-fortiori} results in the accrual of capital
to one company only.

\subsection{Organization of the paper} We now give a brief overview of
the material.

Section \ref{sec: ito process model} introduces an It\^o-process
model for the capitalization of companies.

Perfectly balanced markets and their characterization in terms of
the drifts and volatilities of the capitalization processes are
discussed in section \ref{sec: perfect balance}. To ensure a
non-void discussion, abundance of perfectly balanced markets is
proved.

In section \ref{sec: growth-opt of perf balance} another
economically interesting equivalent formulation of perfectly
balanced markets is established: they achieve maximal growth. With
this characterization, we introduce implied \textsl{shadow interest
rates} in the market.

Next, the concept of balanced markets (a weakening of perfectly
balanced markets) is formulated in exact mathematical terms in
section \ref{sec: balanced econ}. As previously noted, $\Omega$ is
decomposed into $\Omega_b$ and $\Omega_u := \Omega \setminus
\Omega_b$, and we characterize the balanced outcomes event
$\Omega_b$ as the \emph{maximal} set on which an agent who decides
to invest according to any chosen portfolio does not have a chance
to beat the market for any \emph{unbounded} level. In other words,
on $\Omega_b$ agents have a \emph{chance} to beat the market by
specific levels, but this chance is approaching zero
\emph{uniformly} over all portfolios that can be used when the level
becomes arbitrarily large.

The limiting market capital distribution for balanced outcomes is
taken on in Section \ref{sec: limit distr of balanced econ}.
Existence of a limiting capital distribution $\kappa_\infty$ in an
almost sure sense is proved, and under a natural assumption of
company segregation it is shown that all capital will concentrate in
a \emph{single} company and stay there forever.

Easy examples of a simple two-company market are presented in
section \ref{sec: example} that clarify some of the points discussed
previously in the paper.

Finally, in section \ref{sec: qlc} we discuss how all previous
results are still valid in a more general quasi-left-continuous
semimartingale environment (as opposed to a plain It\^o-process
one). Note that, to the best of the author's knowledge, this is the
first time that results on market diversity in such a general
mathematical framework are discussed; in this sense, this last
section is not present just for the sake of abstract generality, but
to ensure that results obtained are not sensitive to the
continuous-semimartingale modeling choice.

\section{The It\^o-Process Model} \label{sec: ito process model}

A continuous semimartingale market model consisting of $d$ different
companies will be consider up to and before section \ref{sec: qlc}.
Actually, attention will be restricted to continuous semimartingales
whose drifts and covariations are absolutely continuous with respect
to Lebesgue measure, It\^o processes being a major example. It shall
be come clear later that this is done only for presentation reasons.

The \textsl{total capitalization} of each company $i = 1, \ldots, d$
is denoted by $S^i$. These capitalizations are modeled as strictly
positive stochastic processes on an underlying probability space
$\probtriple$, adapted to a filtration $\filtration$, assumed
right-continuous and augmented by $\prob$-null sets. The dynamics of
each $S^i$ are:
\begin{equation} \label{eq: model}
\ud S^i_t = S^i_t a^i_t \ud t + S^i_t \ud M^i_t, \ \ \textrm{ for }
i = 1,\ldots, d.
\end{equation}
Here, $a := (a^1, \ldots, a^d)$ is $\bF$-predictable and each $a^i$
represents the \textsl{rate of return} of each company, while $M :=
(M^1, \ldots, M^d)$ is a $(\prob, \bF)$-local martingale for which
we assume that the quadratic covariations satisfy $\ud [M^i, M^j]_t
= c^{i j}_t \ud t$ for a \textsl{local covariation} symmetric matrix
$c := (c^{i j})_{1 \leq i, j \leq d}$, which can be chosen
$\bF$-predictable --- we succinctly write $\ud [M, M]_t = c_t \ud t$
in obvious notation. In order for the model \eqref{eq: model} to
make sense, $a$ and $c$ must satisfy
\[
\int_0^t ( |a^i_u| + c^{ii}_u ) \ud u < \infty, \
\prob\textrm{-a.s.}, \ \ \textrm{ for all } i = 1, \ldots, d
\textrm{ and } t \in \Real_+.
\]

\begin{rem}
Let ``$\Leb$'' denote Lebesgue measure on $\Real_+$ and ``$\det$''
the square-matrix determinant operation. If $\prob[\Leb [\{t \in
\Real_+ |  \det(c_t) \neq 0 \}] = 0]=1$, then there exists a
standard $d$-dimensional $(\bF, \prob)$-Brownian motion $W \equiv
(W^1, \ldots, W^d)$ such that $\ud M_t = \inner{\sigma_t}{\ud W_t}$,
where $\sigma$ is a square root of $c$: $\sigma \sigma^\top = c$
(check for example Karatzas and Shreve \cite{Karatzas-Shreve: BM}).
Then, \eqref{eq: model} is just an It\^o process, and this model is
classic --- see Karatzas and Shreve \cite{Karatzas-Shreve: MMF}. If
$c$ is degenerate on a positive $(\prob \otimes \Leb)$-measure set,
the above representation is still valid if one extends the
probability space. Working directly with \eqref{eq: model} helps to
avoid such complications.
\end{rem}

\begin{rem}
The choice of ``$\ud t$'' above is merely for exposition purposes.
At any rate, in section \ref{sec: qlc} the model is generalized to
the broader class of quasi-left-continuous semimartingales.
\end{rem}

\begin{rem}
It will turn out that it is best to work under the (augmentation of
the) natural filtration generated by $S$, which we denote by
$\bF^S$. Nevertheless, this restriction will \emph{not} be imposed.
Sometimes, we compare obtained results under two filtrations $\bF$
and $\bG$, and it will be assumed that $\mathbf{F}$ \emph{is
contained in} $\mathbf{G}$, in the sense that $\bF \subseteq \bG$,
i.e., $\F_t \subseteq \G_t$ for all $t \in \Real_+$. If $S$ is a
semimartingale of the form (\ref{eq: model}) under $\mathbf{G}$, and
if $\mathbf{F} \supseteq \bF^S$, $S$ is also an $\bF$-semimartingale
and a representation of the form (\ref{eq: model}) is still valid,
with the rates-of-return vector $a$ possibly changed. (The local
covariation matrix $c$ will be the same.)
\end{rem}

\section{Perfectly Balanced Markets} \label{sec: perfect balance}

The significance of perfectly balanced markets has already been
discussed in the Introduction, so here we start directly with their
definition.

\begin{defn}
The \textsl{relative capitalization} $\kappa^i$ of company $i$ is
defined as
\begin{equation} \label{eq: market portf cont}
\kappa^i := \frac{S^i}{S^1 + \ldots + S^d}, \ \ \textrm{ for } i =
1, \ldots, d.
\end{equation}
The market described by (\ref{eq: model}) will be called
\textsl{perfectly balanced} with respect to the probability $\prob$
and the information flow $\bF$ if each $\kappa^i$ is a $(\prob,
\mathbf{F})$-martingale.
\end{defn}

The relative capitalizations process $\kappa := (\kappa^i)_{1 \leq i
\leq d}$ lives in the \textsl{open simplex}
\begin{equation} \label{eq: simplex}
\Delta^{d-1} := \Big\{ x \in \Real^d \such 0 < x^i < 1 \textrm{ and
} \sum_{i=1}^d x^i = 1 \Big\}.
\end{equation}

\begin{rem}
Keep the probability $\prob$ fixed. If the model (\ref{eq: model})
is perfectly balanced with respect some filtration $\bG$ that
contains $\bF$, which in turn contains $\bF^S$, then clearly it is
also perfectly balanced with respect to the information flow
$\mathbf{F}$, since the martingale property remains. The converse
does not necessarily hold: $\mathbf{F}$-perfect balance of the
market does not imply $\mathbf{G}$-perfect balance: the martingale
property might fail when enlarging filtrations. For agents with more
information (political, insider, etc.), the market might fail to
perfectly balance itself.

The weakest form of a perfectly balanced market is obtained when the
filtration is $\bF^S$ --- the one generated by $S$. In fact, an even
smaller filtration can be used, namely, the one generated by
$\kappa$ (since the filtration generated by $S$ has one extra
ingredient, which is the total capitalization $\sum_{j=1}^d S^j$
that disappears when we only consider $\kappa$). It is true that one
can do all subsequent work under this even smaller filtration
--- after all, all that we shall care about is incorporated in
$\kappa$ and if one \emph{starts} by assuming $\kappa$ is the actual
capital process, everything follows.
\end{rem}

\subsection{Characterizing perfectly balanced markets}

Using It\^o's formula and \eqref{eq: model}, it is easily computed
that for all $i = 1, \ldots, d$ we have
\begin{equation} \label{eq: SDE for relcap, general}
\ud \kappa^i_t = \kappa^i_t \inner{\e_i - \kappa_t}{a_t - c_t
\kappa_t} \ud t + \kappa^i_t \inner{\e_i - \kappa_t}{\ud M_t},
\end{equation}
where $\e_i$ the unit vector with all zero entries but the
$i^{\textrm{th}}$, which is unit.

The above equation \eqref{eq: SDE for relcap, general} for
$\kappa^i$, $i=1, \ldots, d$ gives us a way to judge whether the
market is perfectly balanced just by looking at drifts and local
covariations.

\begin{prop} \label{prop: charact of efficiency, cont}
The market is perfectly balanced if and only if there exists a
predictable, one-dimensional process $r$ with $\int_0^t \abs{r_u}
\ud u < + \infty$ for all $t \in \Real_+$, such that, with $\fone$
being the vector in $\Real^d$ will all unit entries: $\fone :=
(1,\ldots, 1)$, we have:
\begin{equation} \label{eq: c rho = b - r fone}
c \kappa = a - r \fone.
\end{equation}
\end{prop}

\begin{proof}
Each of the processes $\kappa^i$ is bounded; therefore it is a
martingale if and only if it is a local martingale, which by view of
\eqref{eq: SDE for relcap, general} will hold if and only if
$\inner{\e_i - \kappa}{a - c \kappa} = 0$. The vector processes
$(\e_i - \kappa)_{1 \leq i \leq d}$ span the linear subspace that is
orthogonal to $\fone$. Thus, in order for $\kappa$ to be a
martingale there should exist a one-dimensional process $r$ such
that $a - c \kappa = r \fone$. The fact that $r$ can be chosen
predictable and locally integrable follows from the fact that both
$c \kappa$ and $a$ have the corresponding properties.
\end{proof}

\begin{rem} \label{rem: interest rate, prelim}
It should be noted here that the process $r$ plays the r\^ole of a
\emph{shadow interest rate} in the market, in the absence of a
banking device that will produce one. To support this claim, suppose
for a minute that one of the companies, say the first, behaves like
a savings account, so that (if only approximately) $S^1$ has only a
``$\ud t$'' component, i.e., $c^{1 i} = 0$ for $i = 1, \ldots, d$.
Multiplying from the left both sides of the relationship (\ref{eq: c
rho = b - r fone}) with the unit vector $\e_1$ we get $a^1 = r$,
i.e., that $r$ is the interest rate. In the absence of a risk-free
company one cannot carry the previous analysis, but an
equilibrium-type argument gives the same conclusion; we come back to
this point in subsection \ref{subsec: interest rate} with a more
thorough discussion.
\end{rem}

Remark  \ref{rem: interest rate, prelim} makes it plausible to
define an \textsl{interest rate process} as being a predictable,
one-dimensional process $r$ with $\int_0^t \abs{r_u} \ud u < +
\infty$ for all $t \in \Real_+$.

\smallskip

The result of Proposition \ref{prop: charact of efficiency, cont}
should be interpreted as a \emph{linear} relationship between the
local covariation and the drifts of the company capitalization
processes, modulo an interest rate process. It is obvious that this
is a \emph{very} restrictive condition; we shall see in Section
\ref{sec: balanced econ} how to weaken it, and we shall discuss how
this softer notion of (not necessarily perfectly) balanced markets
ties with efficiency.

\subsection{Construction of perfectly balanced markets}

Equations  \eqref{eq: SDE for relcap, general} and \eqref{eq: c rho
= b - r fone} combined imply that in a perfectly balanced market the
process $\kappa$ must satisfy the following system of stochastic
differential equations:
\begin{equation} \label{eq: SDE for mu, eff}
\ud \kappa^i_t = \kappa^i_t \inner{\e_i - \kappa_t}{\ud M_t},
\textrm{ for all }  i = 1, \ldots, d.
\end{equation}

The natural question to ask at this point is: \emph{do mathematical
models of perfectly balanced markets exist?} If they do exist,
\eqref{eq: c rho = b - r fone} as well as the stochastic
differential equations \eqref{eq: SDE for mu, eff} must hold. The
following proposition shows that a plethora of perfectly balanced
models exist.

\begin{thm} \label{thm: existence of perf bal econ, cont}
Consider a $d$-dimensional  continuous $(\bF, \prob)$-local
martingale $M$ whose quadratic covariation process satisfies $\ud
[M, M]_t = c_t \ud t$. Then, for any $\F_0$-measurable initial
condition $\kappa_0 \equiv (\kappa_0^i)_{1 \leq i \leq d}$ with
$\prob[\kappa_0 \in \Delta^{d-1}] = 1$ the system of stochastic
differential equations \eqref{eq: SDE for mu, eff} has a unique
strong solution for all $t \in \Real_+$ that lives on
$\Delta^{d-1}$.

Further, for any interest rate process $r$ and $\F_0$-measurable
initial condition $S_0 \equiv (S_0^i)_{1 \leq i \leq d}$ with $S^i_0
/ \sum_{j=1}^{d} S^j_0 = \kappa^i_0$, if we define $a := c \kappa +
r \fone$ and the process $S$ via \eqref{eq: model}, we get a model
of a perfectly balanced market.
\end{thm}

\begin{proof}

The second paragraph of the Proposition's statement is obvious from
our previous discussion; we only need prove that the system of
stochastic differential equations \eqref{eq: SDE for mu, eff} has a
unique strong solution for $t \in \Real_+$ that lives on
$\Delta^{d-1}$.

To begin, consider the unit cube $[0,1]^d$ in $\Real^d$. The
volatility co\"efficients appearing in \eqref{eq: SDE for mu, eff}
are quadratic in $\kappa$, thus are obviously Lipschitz as a
functions of $\kappa$ on $[0,1]^d$; then, the standard theorem on
strong solutions of stochastic differential equations gives that
\eqref{eq: SDE for mu, eff} has a unique strong solution for $t$ in
a stochastic interval $\dbra{0, \tau}$, where $\tau$ is a stopping
time such that for all $t < \tau$ we have $\kappa_t \in (0,1)^d$,
while on $\{ \tau < + \infty \}$ we have $\kappa_\tau \in
\partial [0,1]^d$ (the \textsl{boundary} of $[0,1]^d$). First, it will be shown that $\kappa$ is
$\Delta^{d-1}$-valued on $\dbra{0,\tau}$, and then that $\prob[\tau
= +\infty] = 1$.

Using \eqref{eq: SDE for mu, eff} one can compute that on
$\dbra{0,\tau}$ the process $z := 1 - \inner{\fone}{\kappa}$
satisfies the stochastic differential equation $\ud z_t = - z_t
\inner{\kappa_t}{\ud M_t}$ (observe that now $\kappa$ is known).
Since $z_0 = 0$, the unique strong solution of this last equation is
$z \equiv 0$, so $\inner{\fone}{\kappa} = 1$ and $\kappa$ is
$\Delta^{d-1}$-valued on $\dbra{0,\tau}$.

Now, on $\dbra{0,\tau}$ we have $0 < \kappa^i < 1$ for each
$i=1,\ldots d$. Using It\^o's formula for the logarithmic function
and \eqref{eq: SDE for mu, eff} once again we get for $t \in
\dbraco{0,\tau}$ that
\[
 \log \kappa^i_t = \log \kappa^i_0 - \frac{1}{2} \int_0^t \inner{\e_i - \kappa_u}{c_u (\e_i -
\kappa_u)} \ud u + \int_0^t \inner{\e_i - \kappa_u}{\ud M_u}.
\]
Both the finite-variation part and the quadratic variation of the
local martingale part of the semimartingale $\log \kappa^i$ are
finite on any bounded interval as long as $\kappa \in \Delta^{d-1}$;
it follows that on the event $\{ \tau < +\infty \}$ we have $\lim_{t
\uparrow \tau} \log \kappa_t^i \in \Real$, which implies that
$\lim_{t \uparrow \tau} \kappa^i_t > 0$. Since $\kappa$ is
$\Delta^{d-1}$-valued on $\dbra{0, \tau}$, it also follows that
$\lim_{t \uparrow \tau} \kappa^i_t < 1$ for all $i = 1, \ldots, d$.
This contradicts the fact that we are assumed to work on the event
$\{ \tau < +\infty \}$, therefore $\prob[\tau = + \infty] = 1$.
\end{proof}

\begin{rem}
One of the reasons \emph{not} to require $\bF$ to be the one
generated by $S$ is the constructive Theorem \ref{thm: existence of
perf bal econ, cont}, where we start a-priori with some filtration
$\bF$ that makes $M$ a $\prob$-martingale and $r$ adapted. If
wanted, after the construction of $\kappa$ has been carried out we
can pass from $\bF$ to the generally smaller $\bF^S$.
\end{rem}

\begin{rem}
Apart from its mathematical significance, Theorem \ref{thm:
existence of perf bal econ, cont} also has interesting economic
implications. When writing the dynamics \eqref{eq: model} of a model
we assume that \emph{both} the drift vector $a$ and the local
covariation matrix $c$ are observable. Nevertheless, both in a
statistical and in a philosophical sense, covariances are easier to
assess than drifts. From a statistical point of view, high-frequency
data can lead to reasonably good estimation of $c$ --- and the ideal
case of continuously-collected data leads to perfect estimation.
Nevertheless, there is no easy way to estimate $a$, even if we
assume it is a constant: one has to wait for too long a time to get
any sensible estimate. In a more philosophical sense, economic
agents might not have a complete sense of how the prices will move,
but they might very well have an idea of how risky the companies
are, and how a change in the capitalization of one company would
affect another one, i.e., exactly the local covariation matrix $c$.
To this effect, Theorem \ref{thm: existence of perf bal econ, cont}
implies that simple knowledge of the local covariations $c$, the
interest rate $r$ (see Remark \ref{rem: interest rate, prelim} and
subsection \ref{subsec: interest rate} in this respect) and the
relative capitalizations at time $t = 0$ is enough to provide the
whole process of relative capitalizations; and by this, we also get
the drifts $a$. Thus, in perfectly balanced markets, a good estimate
of $c$ is enough to provide good estimates for the drift $a$ as
well.
\end{rem}

\section{Growth-Optimality of Perfectly Balanced Markets} \label{sec: growth-opt of perf balance}

We discuss here an ``economically optimal'' property of perfectly
balanced markets that actually turns out to be an equivalent
formulation in a sense. We also elaborate on how the process $r$ of
Proposition \ref{prop: charact of efficiency, cont} should be
thought as a shadow interest rate prevailing in the market.

\subsection{Agents and investment}
In a market with $d$ companies whose capitalizations are described
by the dynamics \eqref{eq: model}, we also consider a \emph{savings
account} offered by a bank, described by some interest rate process
$r$. One unit of currency invested in (i.e., loaned to) the bank at
time $s$ will grow to $\int_s^t r_u \ud u$ by time $t > s$. We
remark that existence of a bank does not add wealth to the market
directly, although can do so indirectly by adding more flexibility
to the financial agents --- in other words, the net amount invested
in the bank must be zero: some lend and some borrow, but the total
position should be neutral.

We now discuss the behavior of an individual agent in the market;
this agent decides to invest a portion of the total capital-in-hand
in each of the $d$ companies, and the remaining wealth in the
savings account. We shall be denoting by $\pi^i_t$ the proportion of
the capital invested in the $i^{\textrm{th}}$ company; then, $1 -
\inner{\pi}{\fone}$ proportion of the capital-in-hand is put into
savings. To ensure than no clairvoyance into the future is allowed,
the vector process $\pi := (\pi^i)_{1 \leq i \leq d}$ should be
predictable with respect to the filtration of the individual agent,
which is at least as large as $\bF^S$.

We model portfolio constraints that an agent might be faced with via
a set-valued process $\C$; henceforth we shall be assuming that for
each $(\omega, t) \in \Omega \times \Real_+$:
\begin{enumerate}
  \item $\overline{\Delta}^{d-1} \subseteq \C (\omega, t)$, where $\overline{\Delta}^{d-1}$ is the closure of the open simplex of \eqref{eq: simplex};
  \item the set $\C(\omega, t)$ is closed and convex; and
  \item $\C$ is predictable, in the sense that $\{ (\omega, t) \in \Omega \times \Real_+ \such \C(\omega, t) \cap F \neq \emptyset \}$ is a predictable set for all closed $F \subseteq
  \Real^d$.
\end{enumerate}
Then, a \textsl{$\C$-constrained portfolio} is a predictable,
$d$-dimensional process $\pi$ that satisfies $\pi(\omega, t) \in
\C(\omega, t)$ for all $(\omega, t) \in \Omega \times \Real_+$, and
\begin{equation} \label{eq: integrab for portf}
\int_0^t ( |\inner{\pi_u}{a_u}| + \inner{\pi_u}{c_u \pi_u} ) \ud u <
\infty, \textrm{ for all } t \in \Real_+.
\end{equation}
The class of all $\C$-constrained portfolios is denoted by $\Pi_\C$.

\smallskip

The most important case in the discussion to follow is the most
restrictive case of constraints $\C = \overline{\Delta}^{d-1}$: the
agent has only access to invest in the ``actual'' companies of the
market --- in this case, the bank is not even needed.

The portfolio integrability requirement (\ref{eq: integrab for
portf}) is a technical one, but it is the weakest assumption in
order for the stochastic integrals appearing below in (\ref{eq:
wealth process}) to make sense. The requirement is certainly
satisfied if $\pi$ is $\prob$-a.s. bounded on every interval $[0,t]$
for $t \in \Real_+$ --- for example if $\pi$ is
$\overline{\Delta}^{d-1}$-valued.

\smallskip

The initial investment of an agent at time zero is always normalized
to be a unit of currency. Assuming this and investing according to
$\pi \in \Pi_\C$, the corresponding wealth process $V^\pi$ of the
particular agent is described by $V^\pi_0 = 1$ and
\begin{equation} \label{eq: wealth process}
\frac{\ud V^\pi_t}{V^\pi_t} \ = \ \sum_{i=1}^d \pi_t^i \frac{\ud
S^i_t}{S^i_t} + \Big( 1 - \sum_{i=1}^d \pi_t^i \Big) r_t \ud t \ = \
\big(r_t + \inner{\pi_t}{a_t - r_t \fone}\big) \ud t +
\inner{\pi_t}{\ud M_t}.
\end{equation}

The collective investment of all agents is captured by the
percentage of the total market capitalization invested in each
company, i.e., the relative capitalizations $\kappa = (\kappa^i)_{1
\leq i \leq d}$, which is an $\mathbf{F}$-predictable vector process
(as it is $\bF^S$-adapted and continuous) and satisfies $\kappa \in
\Delta^{d-1}$; thus $\kappa$ can be viewed as a portfolio, and as
such it is called the \textsl{market portfolio}. Here is the reason
for such a name: using (\ref{eq: wealth process}) one checks that
$V^\kappa = \inner{S}{\fone} / \inner{S_0}{\fone}$, where
$\inner{S}{\fone} = \sum_{i=1}^d S^i$ is the total capital of the
market: investing according to $\kappa$ is tantamount to owning the
whole market, relative to the initial investment, which is
normalized to unit.

\subsection{Growth and growth-optimality}

The process $a^\pi := \inner{\pi}{a}$ appearing in \eqref{eq: wealth
process} is known as the \textsl{rate of return} of $V^\pi$; it is
the instantaneous return that the strategy gives on the invested
capital. Nevertheless, for long-time-horizon investments, rates of
return fail to give a good idea of the behavior of the wealth
process. A more appropriate tool for analyzing asymptotic behavior
is the growth rate (see for example Fernholz \cite{Fernholz: SPT}),
which we now define.

For a portfolio $\pi \in \Pi_\C$, its log-wealth process is the
semimartingale $\log V^\pi$. It\^o's formula gives $\ud \log V^\pi_t
= \g^\pi_t \ud t + \inner{\pi_t}{\ud M_t}$, where
\begin{equation} \label{eq: growth rate}
\g^\pi := r + \inner{\pi}{a - r \fone} - \frac{1}{2} \inner{\pi}{c
\pi}
\end{equation}
is the \textsl{growth rate} of the portfolio $\pi$. The portfolio
$\rho \in \Pi_\C$ will be called \textsl{growth-optimal in the
$\C$-constrained class} if
\begin{equation} \label{eq: growth optimal maxim}
\g^{\rho} (\omega, t) = \g^* (\omega, t) := \sup_{\pi \in \C} \g^\pi
(\omega, t), \textrm{ for all } (\omega, t) \in \Omega \times
\Real_+.
\end{equation}
The whole market is called a \textsl{growth market} if the market
portfolio $\kappa$ is growth optimal over all possible portfolios.

\begin{prop} \label{prop: per bal iff growth}
A market described by an interest rate process $r$ and \eqref{eq:
model} for the company capitalizations is a growth market if and
only if $c \kappa = a - r \fone$.
\end{prop}

\proof In order to have a growth market, $\kappa$ must solve the
quadratic problem
\begin{equation} \label{eq: growth quadratic maxim}
\max_\p \Big\{ r + \inner{\p}{a - r \fone} - \frac{1}{2}
\inner{\p}{c \p} \Big\}
\end{equation}
over all $\p \in \Real^d$ where we have hidden the dependence on
$(\omega, t)$. The growth rate function of \eqref{eq: growth rate}
is concave, and first-order conditions imply that in order for
$\kappa$ to be a solution to the optimization problem we must have
$a - c \kappa = r \fone$. \qed

\begin{rem} \label{rem: on numeraire property of growth optimal}
Generalizing a bit the method-of-proof of Proposition \ref{prop: per
bal iff growth}, we can give the following characterization: $\rho$
is $\C$-constrained growth optimal portfolio if and only if $V^\pi /
V^\rho$ is a supermartingale for all $\pi \in \Pi_\C$. Indeed, for
\emph{any} two portfolios $\pi$ and $\rho$, one can use \eqref{eq:
wealth process} and It\^o's formula to get that $V^\pi / V^\rho$ is
a supermartingale is and only if $\inner{\pi - \rho}{a - r \fone - c
\rho} \leq 0$; this is exactly the first-order condition for
maximization of \eqref{eq: growth quadratic maxim} over $\C$.
\end{rem}

\begin{rem}
Statistical tests of the ``perfectly balanced market'' hypothesis
have appeared in the literature in the seventies, where it was
actually tested whether the market portfolio is equal to the
growth-optimal one (the connection is obvious in view of Proposition
\ref{prop: per bal iff growth} --- see also the discussion in the
next subsection \ref{subsec: interest rate}). We mention in
particular the works of Roll \cite{Roll}, as well as Fama and
MacBeth \cite{Fama-McBeth} that treat the New York Stock Exchange as
the ``market''. In both papers, there does not seem to be conclusive
evidence on whether the perfect-balance hypothesis holds or not;
although it cannot be rejected at any reasonably high statistical
significance level, there are noteworthy deviations mentioned
therein.
\end{rem}

\subsection{Interest rate} \label{subsec: interest rate}

Proposition \ref{prop: per bal iff growth} clearly shows the
connection between growth and perfectly balanced markets. The
difference between Propositions \ref{prop: charact of efficiency,
cont} and \ref{prop: per bal iff growth} is that in the former we
\emph{infer} the existence of an interest rate $r$ that satisfies $c
\kappa = a - r \fone$, while in the latter the interest rate process
is given as a market parameter.

In fact, if existence of an interest rate process is not assumed,
and a growth market is defined as one where $\kappa$ maximizes the
growth rate over all portfolios in the constrained set $\C = \{ x
\in \Real^d \such \inner{x}{\fone} = 1 \}$, then going through the
proof of Proposition \ref{prop: per bal iff growth} using
Lagrange-multiplier theory for constrained optimization, the
relationship $c \kappa = a - r \fone$ for \emph{some} interest rate
process $r$ will be inferred again, exactly as in the case of
perfectly balanced markets. Thus, the two concepts of growth and
perfectly balanced markets are identical in this sense.

\smallskip

Now, an equilibrium argument will be used to show that even in the
absence of a bank, the arbitrary process $r$ obtained in the case
where the market is perfectly balanced really plays a r\^ole of an
interest rate. Suppose that all of a sudden, the market decides to
build a bank and has to decide on what interest rate $\tilde{r}$ to
offer. In the next paragraph we answer the following question:
\emph{What should this process $\tilde{r}$ be in order for the
market to stay in perfectly balanced state}? Then, $\tilde{r}$ is an
an equilibrium interest rate.

Before the introduction of a bank the market was perfectly balanced,
i.e., $c \kappa = a - r \fone$ was true for \emph{some}
one-dimensional process $r$. The introduction of a savings account
gives more freedom to individual agents: now they can borrow or lend
at the risk-free interest rate $\tilde{r}$. The ``representative
agent'' in the augmented (with the bank) market will still try to
maximize growth, as before, and for this representative agent the
wealth proportion held in the bank should be zero. Indeed, if in
trying to maximize the growth rate the representative agent found
that the optimal holdings in the risk-free security is positive, the
overall feeling of the agents is that the interest rate level
$\tilde{r}$ is attractive for saving, and more agents would be
inclined to save money that to borrow for investment in the riskier
company of the market; this would create instability because supply
for funds to be invested in riskier companies would exceed demand.
The exact opposite of what was just described would happen if the
representative agent's optimal holdings in the risk-free security
were negative. Proposition \ref{prop: per bal iff growth} implies
that after the introduction of a bank we should have $c \kappa = a -
\tilde{r} \fone$; nevertheless, just before the bank appeared we had
$c \kappa = a - r \fone$. The only way that both can hold is $r =
\tilde{r}$, which shows that $r$ really plays the r\^ole of an
equilibrium interest rate process, even in the absence of a bank.

\section{Balanced Markets} \label{sec: balanced econ}

The definition of a perfectly balanced market is restrictive, since
the martingale property for the relative capitalizations is not
expected to exactly hold. Sometimes it might fail and it also might
take some time to return to equilibrium, as explained in the
Introduction. We therefore want to say that the market will be
balanced (though not necessarily perfectly) if the martingale
property holds only ``approximately''. No such reasonable notion
exists, and one needs to work around it. In this section we
elaborate on balanced markets and their close relation to a concept
of ``efficiency''.

\subsection{Formal definitions}

According to Proposition \ref{prop: per bal iff growth} and the
content of subsection \ref{subsec: interest rate}, a market equipped
with a bank is perfectly balanced if and only if $\g^\kappa = \g^*$,
where $\g^* \equiv \g^{* (\bF, \C)}$ is the maximal growth that can
be obtained by using $\bF$-predictable and $\C$-constrained
portfolios. In general, we have $\g^\kappa \leq \g^*$, and the
market will be balanced if this difference is not very large.

\begin{defn} \label{dfn: balanced econ}
For some filtration $\bF$ and constraints set $\C$, define the
continuous and increasing \textsl{loss-of-perfect-balance} process
$L$ via
\[
L_t \ \equiv \ L^{\bF, \C}_t \ := \ \int_0^t (\g^{* (\bF, \C)}_u -
\g_u^\kappa) \ud u,
\]
and write $\Omega = \Omega_b \cup \Omega_u$, where $\Omega_b  \equiv
\Omega^{\bF, \C}_b := \{ L^{\bF, \C}_\infty < + \infty \}$ are the
\textsl{balanced} outcomes and $\Omega_u \equiv \Omega^{\bF, \C}_u
:= \{ L^{\bF, \C}_\infty = + \infty \} = \Omega \setminus \Omega_b$
the \textsl{unbalanced} outcomes.

If $\prob[\Omega^{\bF, \C}_b] = 1$, the market described by
(\ref{eq: model}) will be called \textsl{balanced} with respect to
the probability $\prob$, the information flow $\bF$ and the
constraints $\C$.
\end{defn}

\begin{rem} \label{rem: max growth vs growth optimal}
If a predictable process $\rho$ that solves the maximization problem
\eqref{eq: growth optimal maxim} exists for all $(\prob \otimes
\Leb)$-almost every $(\omega, t) \in \Omega \times \Real_+$ and
$\rho$ satisfies the integrability conditions \eqref{eq: integrab
for portf} we then have $\g^* = \g^\rho$. This always happens if
$\C$ is contained in a fixed compact subset $K$ of $\Real^d$ for all
$(\omega, t) \in \Omega \times \Real_+$.

In general, a predictable process $\rho$ solving \eqref{eq: growth
optimal maxim} might not exist; even if it does exist, the
integrability conditions \eqref{eq: integrab for portf} might not be
fulfilled. It can be shown that $\rho$ exists and satisfies
\eqref{eq: integrab for portf} if and only if $L_t < \infty$ for all
$t \in \Real_+$, $\prob$-a.s. A thorough discussion of these points
is made in Karatzas and Kardaras \cite{KK: num and arbitrage}.
\end{rem}

Consider two filtrations $\bF$ and $\mathbf{G}$ such that $\bF^S
\subseteq \bF \subseteq \bG$; $\mathbf{G}$-perfect balance implies
$\mathbf{F}$-perfect balance. The same holds for simply balanced
markets.

\begin{prop} \label{prop: comparison of balanced}
Consider two pairs of filtrations and constraints $(\bF, \C)$ and
$(\bG, \K)$ with $\bF^S \subseteq \bF \subseteq \bG$ and $\C
\subseteq \K$. We then have $\g^{* (\bF, \C)} \leq \g^{* (\bG,
\K)}$; as a consequence, $L^{\bF, \C} \leq L^{\bG, \K}$ and
$\Omega_b^{\bG, \K} \subseteq \Omega_b^{\bF, \C}$.
\end{prop}

\begin{proof}

For all $n \in \Natural$ set $\C_n := \C \cap [-n, n]^d$ and $\K_n
:= \K \cap [-n, n]^d$. We then have that $\lim_{n \to \infty}
\uparrow \g^{* (\bF, \C_n)} = \g^{* (\bF, \C)}$ and $\lim_{n \to
\infty} \uparrow \g^{* (\bG, \K_n)} = \g^{* (\bG, \K)}$ and thus it
suffices to prove $\g^{* (\bF, \C)} \leq \g^{* (\bG, \K)}$ under the
assumption $\C \subseteq \K \subseteq K$ for some compact set $K$.
According to Remark \ref{rem: max growth vs growth optimal}, under
this assumption the growth-optimal portfolios $\rho(\bF, \C)$ and
$\rho(\bG, \K)$ exist and $\g^{* (\bF, \C)} = \g^{\rho (\bF, \C)}$
as well as $\g^{* (\bG, \K)} = \g^{\rho (\bG, \K)}$. From Remark
\ref{rem: on numeraire property of growth optimal} we know that
$V^{\rho (\bF, \C)} / V^{\rho (\bG, \K)}$ is a positive
supermartingale, which gives that $\log(V^{\rho (\bF, \C)} / V^{\rho
(\bG, \K)})$ is a local supermartingale; the drift of the last local
supermartingale --- which should be decreasing --- is $\int_0^\cdot
(\g_t^{\rho (\bF, \C)} - \g_t^{\rho (\bG, \K)}) \ud t$, which gives
us $\g^{* (\bF, \C)} \leq \g^{* (\bG, \K)}$ and completes the proof.
\end{proof}
\subsection{Some discussion} We contemplate slightly on balanced markets.

\subsubsection{Trivial example}
Perfectly balanced markets satisfy $L \equiv 0$, and are therefore
balanced.

\subsubsection{No bank}
Let us assume now that $\C = \{ x \in \Real^d \ | \ \inner{x}{\fone}
= 1\}$ --- we are allowed to invest in the risky companies, but
there is no bank (for us, at least).

We assume that $c$ is non-degenerate for $(\prob \otimes
\Leb)$-almost every $(\omega, t) \in \Omega \times \Real_+$; then,
the maximization problem \eqref{eq: growth optimal maxim} has a
solution $\rho$ that satisfies $c \rho = a - r \fone$ for some
\emph{unique} one-dimensional process $r$. On the $(\prob \otimes
\Leb)$-full measure subset of $\Omega \times \Real_+$ where $c$ is
non-singular it is clear that $\rho = c^{-1} (a - r \fone)$; using
$\inner{\rho}{\fone} = 1$ it is easy to see that
\begin{equation} \label{eq: interest rate implied}
r \ = \ \frac{\inner{a}{c^{-1} \fone} - 1}{\inner{\fone}{c^{-1}
\fone}}
\end{equation}
Now, straightforward computations give
\[
\g^* - \g^\kappa \ \equiv \ \g^\rho - \g^\kappa \ = \ \frac{1}{2}
\inner{\kappa - \rho}{ c (\kappa - \rho)} \ = \ \frac{1}{2} \big|
c^{-1/2} (c \kappa - a + r \fone) \big|^2.
\]
(One can also show the last relationship observing that $V^\kappa /
V^\rho$ is a local martingale and taking the logarithm.) Perfectly
balanced markets satisfy $c \kappa - a + r \fone = 0$ identically
with $r$ given by \eqref{eq: interest rate implied}; simply balanced
markets do not satisfy the last equation identically, but
approximately: $\int_0^\infty | c_t^{-1/2} (c_t \kappa_t - a_t + r_t
\fone) |^2 \ud t < \infty$.

\smallskip

We remark that on $\Omega_b$, $r$ earns the name of an interest rate
process, i.e., it is locally integrable. More specifically, it will
be shown below that for any random time $\tau$ we have $\int_0^\tau
|r_u| \ud u < \infty$ on $\{ \tau < \infty, \ L_\tau < \infty \}$.
Define $F_t := \int_0^t \inner{\kappa_u}{c_u \kappa_u} \ud u$; on
$\{ \tau < \infty \}$ we have $F_\tau < \infty$. The Cauchy-Schwartz
inequality gives
\[
\int_0^\tau | \inner{\kappa_u}{c_u \kappa_u} - \inner{\kappa_u}{a_u}
- r_u | \ud u = \int_0^\tau |\inner{\kappa_u}{c_u (\kappa_u -
\rho_u)} | \ud u \leq \sqrt{L_\tau F_\tau} < \infty,
\]
on $\{ \tau < \infty, \ L_\tau < \infty \}$. Since on $\{ \tau <
\infty \}$ we have $\int_0^\tau |\inner{\kappa_u}{c_u \kappa_u} -
\inner{\kappa_u}{a_u}| \ud u < \infty$; we conclude that
$\int_0^\tau | r_u | \ud u < \infty$ on $\{ \tau < \infty, \ L_\tau
< \infty \}$, as proclaimed.

\subsubsection{Interest rate revisited}

Continuing the above discussion, where no bank is present, suppose
that we wish to introduce an interest rate process $\tilde{r}$ in
such a way as to keep the market balanced
--- at least on the event that it was balanced before. In the case of
perfectly balanced market, $\tilde{r} \equiv r$ must hold
--- here, we shall see that we have this last equality holding
\emph{approximately}.

We still assume that $c$ is non-singular on a set of full $(\prob
\otimes \Leb)$-measure (which is very reasonable to justify the
introduction of a bank). A solution $\tilde{\rho}$ of the
optimization problem \eqref{eq: growth optimal maxim} in the market
augmented with the bank exists, and $c \tilde{\rho} = a - \tilde{r}
\fone$. Straightforward, but somewhat lengthy, computations show
that
\[
\int_0^\infty (\g_t^{\tilde{\rho}} - \g_t^\kappa) \ud t \ = \
\int_0^\infty (\g_t^{\rho} - \g_t^\kappa) \ud t + \frac{1}{2}
\int_0^\infty \inner{\fone}{c_t^{-1} \fone} |\tilde{r_t} - r_t|^2
\ud t,
\]
where $\rho := c^{-1} (a - r \fone)$ and $r$ is given by \eqref{eq:
interest rate implied}. Introducing a bank that offers interest rate
$\tilde{r}$ keeps the market balanced if and only if $\int_0^\infty
\inner{\fone}{c_t^{-1} \fone} |\tilde{r}_t - r_t|^2 \ud t < \infty$,
which can be seen as an approximate equality between $r$ and
$\tilde{r}$.

\subsubsection{Equivalent Martingale Measures} \label{rem: on EMM}
We now delve into the relationship between balanced markets and the
existence of a probability $\qprob \sim \prob$ that makes the
relative capitalizations $\kappa^i$ $\qprob$-martingales. We call
such a probability $\qprob$ an equivalent martingale measure (EMM),
although it does not apply directly to the \emph{actual}, rather to
the \emph{relative} capitalizations. The concept of a balanced
market is closely related, but weaker than the existence of an EMM.
It is not hard to see why it is weaker: assume the existence of an
EMM $\qprob$ and denote by $Z$ the density process, i.e., $Z_t :=
(\ud \qprob / \ud \prob) |_{\F_t}$. Since $\qprob \sim \prob$, we
have $\prob[Z_\infty] > 0$. The Kunita-Watanabe decomposition
implies
\[
Z_t = E_t N_t, \textrm{ with } E_t := \exp \pare{\int_0^t
\inner{h_u}{\ud M_u} - \frac{1}{2} \int_0^t \inner{h_u}{c_u h_u} \ud
u}
\]
where $h$ is an $d$-dimensional predictable process and the strictly
positive local martingale $N$ is strongly orthogonal to $M$. The
integrand $h$ need not be unique, but the local martingale
$\int_0^\cdot \inner{h_u}{\ud M_u}$ is. Since $\kappa$ has to be a
$\qprob$-martingale, one can show that we can choose $h = \rho -
\kappa$, where $\rho$ is the growth-optimal portfolio, that
\emph{must} exist. Since $Z_\infty
> 0$ and $N_\infty < + \infty$, $\prob$-a.s. we also have that
$E_\infty > 0$, $\prob$-a.s.; in view of Lemma \ref{lem: asymptotic
pos loc mart, cont} this is equivalent to saying that the quadratic
variation of the local martingale $\int_0^\cdot \inner{h_u}{\ud
M_u}$ is finite at infinity --- but this is exactly $L_\infty$; thus
the existence of an EMM implies that the market is balanced.

In section \ref{sec: example} we shall see by example that the
notion of a balanced market is actually \emph{strictly} weaker than
existence of an EMM $\qprob$.

\subsection{Balanced markets and efficiency}

The chances for an agent to do well relatively to the overall wealth
are very different depending on which of the events $\Omega_b$ and
$\Omega_u$ is being considered. The next result gives a
characterization of $\Omega_u$ in terms of beating the whole market.

\begin{thm} \label{thm: charact of weak eff}
We consider the model \eqref{eq: model} valid under some filtration
$\bF \supseteq \bF^S$ and the constraints set $\C =
\overline{\Delta}^{d-1}$.

\noindent $\bullet$ On $\Omega_b$, and for any portfolio $\pi \in
\Pi_\C$ the limit of the relative wealth process $\lim_{t \to
\infty} (V^\pi_t / V^\kappa_t)$ exists and is $\Real_+$-valued. The
probability of beating the whole market for ever-increasing levels
converges to zero uniformly among all portfolios:
\begin{equation} \label{eq: bddness in prob on Omega_e}
\lim_{m \to \infty} \downarrow \sup_{\pi \in \Pi_\C} \prob \Big[
\sup_{t \in \Real_+} \Big( \frac{V^\pi_t}{V^\kappa_t} \Big) > m \
\big{|} \ \Omega_b \Big] = 0.
\end{equation}

\noindent $\bullet$ Further, $\Omega_b$ is the maximal set that
(\ref{eq: bddness in prob on Omega_e}) holds: there exists $\rho \in
\Pi_\C$ such that $\lim_{t \to \infty} ( V^{\rho}_t / V^\kappa_t )$
is $(0, \infty]$-valued, and $\Omega_u = \{ \lim_{t \to \infty} (
V^{\rho}_t / V^\kappa_t ) = \infty \}$.
\end{thm}

\begin{proof}

Consider the growth optimal portfolio $\rho$ in the class $\Pi_\C$
--- since $\C$ is a constant compact subset of $\Real^d$ this
certainly exists. Take now any portfolio $\pi \in \Pi_\C$; Remark
\ref{rem: on numeraire property of growth optimal} gives that the
relative wealth process $V^\pi / V^\rho$ is a positive
supermartingale. Then, for any $l > 0$ we have $\prob[\sup_{t \in
\Real_+} (V^\pi_t / V^\rho_t) > l] \leq l^{-1}$, i.e., the
collection $\{\sup_{t \in \Real} (V^\pi_t / V^\rho_t) \}_{\pi \in
\Pi_\C}$ is bounded in probability. Further, It\^o's formula for the
semimartingale $\log(V^\kappa / V^\rho)$ reads
\begin{equation} \label{eq: log rel wealth of market to growth opt}
\log \frac{V^\kappa_t}{V^\rho_t} \ = \ - L_t + \int_0^t
\inner{\kappa_u - \rho_u}{ \ud M_u}.
\end{equation}
Observe then that on $\Omega_b$ both the finite-variation part and
the quadratic variation of the local martingale part of the
semimartingale $\log (V^\kappa / V^\rho)$ are finite all the way to
infinity, thus $\inf_{t \in \Real_+} (V^\kappa_t / V^\rho_t) \in (0,
+ \infty)$. Writing $V^\pi / V^\kappa = (V^\pi / V^\rho) (V^\rho /
V^\kappa)$ for all $\pi \in \Pi_\C$, we see that the collection $\{
\sup_{t \in \Real_+} (V^\pi_t / V^\kappa_t) \}_{\pi \in \Pi_\C}$ is
bounded in probability on $\Omega_b$, which is exactly the first
claim (\ref{eq: bddness in prob on Omega_e}).

The fact that $L$ dominates twice the quadratic variation process of
the local martingale $\int_0^\cdot \inner{\kappa_u - \rho_u}{ \ud
M_u}$ enables one to use the strong law of large numbers (Lemma
\ref{prop: strong law of large numbers cont mart} of the Appendix)
in \eqref{eq: log rel wealth of market to growth opt} and show that
we have
\[
\lim_{t \to \infty} \frac{\log(V^\rho_t/V^\kappa_t)}{L_t} = 1,
\]
on $\Omega_u = \{ L_\infty = \infty\}$, which proves the second
claim.
\end{proof}

\begin{rem} \label{reb: extension of effic thm}
The assumption $\C = \overline{\Delta}^{d-1}$ in Theorem \ref{thm:
charact of weak eff} is being made to ease the proof, and also
because it will be the only case we need in the sequel. This
assumption can be dropped; Theorem \ref{thm: charact of weak eff}
still holds, with some possible slight changes which we now
describe.

The essence of the assumption $\C = \overline{\Delta}^{d-1}$ was to
make sure that the growth-optimal portfolio $\rho$ exists in the
class $\Pi_\C$; thus, the proof remains valid whenever $\C$ is
contained in a compact subset of $\Real^d$. In the general case, one
might not be able to use $\rho$ directly (since it might not even
exist), but rather a subsequence of $(\rho_n)_{n \in \Natural}$
where $\rho_n$ defined to be the $\C_n$-constrained growth-optimal
portfolio where $\C_n := \C \cap [-n, n]^d$ and replace the second
bullet in Theorem \ref{thm: charact of weak eff} by

\noindent $\bullet$ Further, $\Omega_b$ is the maximal set that
(\ref{eq: bddness in prob on Omega_e}) holds: if $\prob[\Omega_u] >
0$ one can find a sequence of portfolios $(\rho_n)_{n \in \Natural}$
such that $\lim_{t \to \infty}
\pare{V^{\rho_n}_t / V^\kappa_t}$ exists and
\[
\lim_{n \to \infty} \prob \Big[ \lim_{t \to \infty}
\frac{V^{\rho_n}_t}{V^\kappa_t}
> n \ \big{|} \ \Omega_u \Big] = 1.
\]

\end{rem}

\section{Segregation and Limiting Capital Distribution of Balanced Markets} \label{sec: limit distr of balanced econ}

Here, we describe the limiting behavior of the market on the set of
balanced outcomes $\Omega_b$. We take the latter event to be as
large as possible, which by Proposition \ref{prop: comparison of
balanced} means that for this section we consider the case where the
filtration is $\bF^S$ and $\C = \overline{\Delta}^{d - 1}$. By
Theorem \ref{thm: charact of weak eff}, on the event $\Omega_u$ an
investor with minimal information can construct an all-long
portfolio that can beat the market unconditionally; to keep our
sanity, it is best to assume that the market is balanced.

\subsection{Limiting capital distribution} The following result
is a simple corollary of Theorem \ref{thm: charact of weak eff}.
(\emph{All} set-inclusions appearing from now on are valid modulo
$\prob$.)

\begin{prop} \label{prop: lim cap dist exists for balanced}
$\Omega_b \ \subseteq \ \{ \kappa_\infty := \lim_{t \to \infty}
\kappa_t \textrm{ exists} \}$.
\end{prop}

\proof Write $\kappa^i = \kappa^i_0 (V^{\e_i}/ V^\kappa)$ and use
the first claim of Theorem \ref{thm: charact of weak eff} with $\pi
= \e_i$. \qed

\medskip

Thus, we know that on $\Omega_b$ there exists a limiting capital
distribution in a very strong sense: there is almost sure
convergence of the relative capitalizations vector. The next task is
to identify this distribution.

\subsection{Sector equivalence and segregation}
We give below a definition of some sort of distance between two
companies. To introduce the definition and get an idea of what it
means, remember that if $\pi_1$ and $\pi_2$ are two portfolios, the
drift of the log-wealth process $\log(V^{\pi_1}/V^{\pi_2})$ is
$\int_0^\cdot \g_t^{\pi_1 | \pi_2} \ud t$, where $\g^{\pi_1 | \pi_2}
:= \g^{\pi_1} - \g^{\pi_2}$, and that its quadratic variation is
$\int_0^\cdot c_t^{\pi_1 | \pi_2} \ud t$ where $c^{\pi_1 | \pi_2} :=
\inner{\pi_2 - \pi_1}{c (\pi_2 - \pi_1)}$. In the case where the
portfolios are unit vectors $\e_i$, $\e_j$ for some $1 \leq i,j \leq
d$ we write $\g^{i | j}$ and $c^{i | j}$ for $\g^{\e_i | \e_j}$ and
$c^{\e_i | \e_j}$ respectively.

\begin{defn} \label{dfn: dist of sect con}
Say that two companies $i$ and $j$ in the market are
\textsl{equivalent} (on the outcome $\omega$) and denote $i \sim j$
(more precisely $i \sim_\omega j$) if their total distance
\begin{equation} \label{eq: dist of companys}
d^{i|j} := \int_0^\infty \big( |\g_t^{i|j}| + \frac{1}{2} c_t^{i|j}
\big) \ud t,
\end{equation}
satisfies $d^{i | j} (\omega) < \infty$, and write $i \nsim j$ ($i
\nsim_\omega j$ is more precise) if $d^{i|j} (\omega) = \infty$.

The \textsl{segregation} event is $\Sigma := \{  i \nsim j, \textrm{
for all pairs of companies } (i,j) \}$; if $\prob[\Sigma] = 1$, the
market will be called \textsl{segregated}.
\end{defn}

Market segregation is conceptually very natural. Indeed, if two
companies satisfy $i \sim_\omega j$ for some outcome $\omega \in
\Omega$, then the total quadratic variation of the difference of
their returns all the way to infinity is finite; in this sense, the
total cumulative uncertainty (up to infinity) that they bear is
\emph{very} comparable. The same happens for their growth rates, as
\eqref{eq: dist of companys} implies. In this case they should
really be viewed and modeled as the same entity of the market. To
really speak of ``different'' companies, they must have some
different uncertainty \emph{or} growth characteristics; this makes
Definition \ref{dfn: dist of sect con} perfectly reasonable.

\begin{rem} \label{rem: equiv of portfolios}
An equivalence relation between portfolios $\pi_1$ and $\pi_2$ can
similarly be defined, by postulating that $\pi_1 \sim_\omega \pi_2$
if $\int_0^\infty (|\g_t^{\pi_1|\pi_2}| + c_t^{\pi_1|\pi_2}/2) \ud t
< \infty$. Then, we can write $\Omega_b = \{ \rho \sim \kappa \}$.
To wit, remember that $\Omega_b = \{ \int_0^\infty \g_t^{ \rho |
\kappa } \ud t < \infty \}$, so we certainly have $\{ \rho \sim
\kappa \} \subseteq \Omega_b$. On the other hand, since $V^\kappa /
V^\rho$ is a supermartingale, it is easy to see that we have $2
\g^{\rho | \kappa} \geq c^{\rho | \kappa}$, which gives $\{
\int_0^\infty \g_t^{ \rho | \kappa } \ud t < \infty \} \subseteq \{
\int_0^\infty c_t^{ \rho | \kappa } \ud t < \infty \}$, and thus
$\Omega_b = \{ \rho \sim \kappa \}$.
\end{rem}

It should be clear that
\begin{equation} \label{eq: equiv implies limits}
\{ i \sim j \} \, \subseteq \, \Big\{ \lim_{t \to \infty} \Big( \log
\frac{\kappa^i_t}{\kappa^j_t} \Big) \textrm{ exists} \Big\} \, = \,
\Big\{ \lim_{t \to \infty}  \frac{\kappa^i_t}{\kappa^j_t} \textrm{
exists and is strictly positive} \Big\}.
\end{equation}

\begin{rem} \label{rem: justifying equivalence appel}
The relationship $\sim$ of Definition \ref{dfn: dist of sect con} is
an equivalence relationship. Indeed, suppose that $i$, $j$ and $k$
are three companies. That $i \sim i$ is evident, since $\g^{i|i} =
c^{i | i} = 0$; also, $i \sim j\Leftrightarrow j \sim i$ follows
because $|\g^{i|j}|$ and $c^{i | j}$ are symmetric in $(i, j)$.
Finally, if $i \sim j$ and $j \sim k$, the triangle inequality
$|\g^{i | k}| \leq |\g^{i | j}| + |\g^{j | k}|$ gives $\int_0^\infty
|\g_t^{i | k}| \ud t < \infty$. By It\^o's formula,
\[
\log \Big( \frac{\kappa_t^i}{\kappa_t^k} \Big) \ =  \ \log \Big(
\frac{\kappa_0^i}{\kappa_0^k} \Big) + \int_0^t \g^{i|k}_u \ud u +
\inner{\e_i - \e_k}{M_t}.
\]
Then, $\inner{\e_i - \e_k}{M} = \log(\kappa^i / \kappa^k) -
\log(\kappa_0^i / \kappa_0^k) - \int_0^\cdot |\g_t^{i | k}| \ud t$
is a local martingale. We have $\{ i \sim j \} \cap \{ j \sim k \}
\subseteq \{ \lim_{t \to \infty} \log (\kappa^i_t / \kappa^k_t)
\textrm{ exists} \}$ from \eqref{eq: equiv implies limits}, hence
$\inner{\e_i - \e_k}{M}$ has a finite limit at infinity on $\{ i
\sim j \} \cap \{ j \sim k \}$, which means that its quadratic
variation up to infinity has to be finite on the latter event, i.e.,
$\int_0^\infty |c_t^{i|k}| \ud t < \infty$ on $\{ i \sim j \} \cap
\{ j \sim k \}$, and the claim is proved. The same holds for the
relationship $\sim$ described in Remark \ref{rem: equiv of
portfolios} above for portfolios.
\end{rem}

On the event $\{ \kappa_\infty := \lim_{t \to \infty} \kappa_t
\textrm{ exists} \} \cap \{ i \sim j \}$ we have $\kappa^i_\infty =
0 \Leftrightarrow \kappa^j_\infty = 0$, and thus also
$\kappa^i_\infty
> 0 \Leftrightarrow \kappa^j_\infty > 0$; this is trivial in view of \eqref{eq: equiv implies
limits}. A somewhat surprising partial converse to this last
observation is given now.

\begin{lem} \label{lem: key for divers failure}
For any pair $(i, \, j)$, we have $\Omega_b \cap \{ \kappa^i_\infty
> 0, \ \kappa^j_\infty
> 0 \} \subseteq \{ i \sim j \}$.
\end{lem}

\proof Since $V^\kappa/ V^\rho$ has a strictly positive limit at
infinity on $\Omega_b$, we get that the local martingale $V^{\e^i} /
V^\rho$ has a strictly positive limit at infinity on $\Omega_b \cap
\{ \kappa^i_\infty > 0 \}$. According to Lemma \ref{lem: asymptotic
pos loc mart, cont}, this means that $\int_0^\infty |\g_t^{i |
\rho}| \ud t = \int_0^\infty \g_t^{\rho | i} \ud t < \infty$. Then,
on $\Omega_b \cap \{ \kappa^i_\infty > 0, \ \kappa^j_\infty > 0 \}$
we have both $\int_0^\infty |\g_t^{i | \rho}| \ud t < \infty$ and
$\int_0^\infty |\g_t^{j | \rho}| \ud t < \infty$; since $|\g^{i|j}|
\leq |\g^{i|\rho}| + |\g^{j|\rho}|$, we get $\int_0^\infty |\g_t^{i
| j}| \ud t < \infty$. Now, the fact that $\lim_{t \to \infty}
\log(\kappa^i_t / \kappa^j_t)$ exists on $\{ \kappa^i_\infty
> 0, \ \kappa^j_\infty > 0 \}$ allows us to proceed as in Remark \ref{rem: justifying equivalence
appel} and show that $\int_0^\infty c_t^{i | j} \ud t < \infty$. We
conclude that $i \sim j$ on $\Omega_b \cap \{ \kappa^i_\infty > 0, \
\kappa^j_\infty > 0 \}$. \qed

\subsection{One company takes all} Now comes the main
result of this section.
\begin{thm} \label{thm: one company takes all}
$\Omega_b \cap \Sigma \ \subseteq \ \big \{ \kappa_\infty \in \{
\e_1, \ldots, \e_d \} \big \}$. In particular, in a balanced and
segregated market, $\kappa_\infty$ exists $\prob$-a.s. and is equal
to a unit vector.
\end{thm}

\proof This is a simple corollary of Lemma \ref{lem: key for divers
failure}: On $\Omega_b$, if we had $\kappa^i_\infty > 0$ and
$\kappa^j_\infty > 0$ for any two companies $i$ and $j$, we should
have $i \sim j$; but the segregation event $\Sigma$ is exactly the
one where $i \nsim j$ for all pairs of companies $(i,j)$. \qed

\begin{rem}

This is a follow-up to the discussion in paragraph \ref{rem: on EMM}
on Equivalent Martingale Measures. Existence of an EMM $\qprob$,
coupled with Theorem \ref{thm: one company takes all}, imply that
for each $i \in \{ 1, \ldots, d \}$ we have $\qprob[\kappa^i_\infty
=1 \such \F_0] = \kappa^i_0 > 0$, thus $\prob[\kappa^i_\infty =1
\such \F_0]
> 0$ as well. This ceases to be true anymore if we consider balanced
markets. Indeed, in the next section one finds an example of a
balanced and segregated market, such that a \emph{specific} company
takes over the whole market with probability one.

\end{rem}

\section{Examples} \label{sec: example}

We consider here a parametric ``toy'' market model in order to
illustrate the results of the previous subsections and to clarify
some points discussed. The market will consist of two companies, and
their capitalizations are $S^0$ and $S^1$. Under $\prob$, $S^0
\equiv 1$, while $S^1_0 = 1$ and $\ud S^1_t = S^1_t (a_t \ud t +
\sigma_t \ud W_t)$, where $a$ and $\sigma$ are predictable
processes, $\sigma$ is strictly positive, and $W$ is a
one-dimensional Brownian motion. In the three cases we consider
below we always have $0 \leq a / \sigma^2 \leq 1/2$; it then turns
out that $\rho = (1 - a / \sigma^2, a / \sigma^2)$ and easy
computations show that
\begin{equation} \label{eq: example computation}
L_\infty = \frac{1}{2} \int_0^\infty \Big| \frac{a_t}{\sigma_t} -
\sigma_t \kappa^1_t \Big|^2 \ud t, \textrm{ as well as }  \{ 0 \nsim
1\} = \Big\{ \frac{1}{2} \int_0^\infty |\sigma_t|^2 \ud t = \infty
\Big \}.
\end{equation}

\subsection{Case $a = 0$}
Here, $L_\infty =  (1/2) \int_0^\infty
|\sigma_t \kappa^1_t|^2 \ud t \leq (1/2) \int_0^\infty |\sigma_t
S^1_t|^2 \ud t$. Observe that the random quantity $\int_0^\infty |\sigma_t S^1_t|^2 \ud
t$ is the quadratic variation of the local martingale $S^1$ up to
infinity, which should be finite, since $S^1$ has a limit at
infinity. Therefore, the market is balanced: $\Omega_b = \{ L_\infty < \infty \} =
\Omega$.

Observe also that $\{ 0 \nsim 1\} = \{ \lim_{t \to \infty} S^1_t = 0
\} = \{ \kappa_\infty = \e_0 \}$. Here, the limit in the event
$\Omega_b \cap \{ 0 \nsim 1\} = \{ 0 \nsim 1\}$ is identified as
being $\e_0$, and one sees that on $\{ 0 \sim 1\}$ we have $0 <
\kappa^0_\infty < 1$ as well as $0 < \kappa^1_\infty < 1$. In a
balanced market with equivalent companies the limiting capital
distribution might not be trivial.

Assume now that $\prob[\int_0^\infty |\sigma_t|^2 \ud t = \infty] =
1$; easy examples of this is when $\sigma$ is a positive constant,
or when $S^1$ is the inverse of a three-dimensional Bessel process.
From the discussion above, the market is balanced and segregated. We
note that there \emph{cannot} exist any probability measure $\qprob
\sim \prob$ such that $\kappa$ is a $\qprob$-martingale; for if
there existed one, the bounded martingale $\kappa^1$ would be
uniformly integrable, so that $0 = \expec^\qprob[\kappa^1_\infty] =
\kappa^1_0 = 1/2$ should hold, which is impossible.

This example clearly shows that balanced markets form a strictly
larger class than the ones satisfying the EMM hypothesis discussed
in \ref{rem: on EMM}.

\subsection{Case $\epsilon \leq a / \sigma^2 \leq 1/2 - \epsilon$} Here we assume
the previous inequality holds for all $(\omega,t) \in \Omega \times
\Real_+$ for some $0 < \epsilon < 1/4$; for example, one can just
pick some predictable, strictly positive process $\sigma$ and then
set $a = \sigma^2/4$.

As in the previous case $a = 0$, we have $\{ 0 \nsim 1\} = \{
\lim_{t \to \infty} S^1_t = 0 \} = \{ \kappa^1_\infty = 0 \}$ ---
this follows from \eqref{eq: example computation}; just divide the
equality
\[
\log S^1_t =  \int_0^t \Big( a_u - \frac{1}{2} \sigma_u^2 \Big) \ud
u + \int_0^t \sigma_u \ud W_u
\]
by $\int_0^t |\sigma_u|^2 \ud t$ and then use $a - \sigma^2/2 \leq -
\epsilon \sigma^2$ as $t$ tends to infinity. Because of this last
fact, using $\epsilon \leq a / \sigma^2$ and \eqref{eq: example
computation} again, we easily get $\{ 0 \nsim 1 \} \subseteq \{
L_\infty = \infty\} = \Omega_u = \Omega \setminus \Omega_b$. This
example shows that the limiting capital distribution can be
concentrated in one company even in the set where then market is not
balanced.

\subsection{Case $a = \sigma^2/2$} In this case, $\log S^1$ is a local martingale
with quadratic variation $\int_0^\cdot |\sigma_t|^2 \ud t$, and thus
on $\{0 \nsim 1\} = \{ \int_0^\infty |\sigma_t|^2 \ud t = \infty \}$
we have $\liminf_{t \to \infty} \kappa^1_t = 0$ and $\limsup_{t \to
\infty} \kappa^1_t = 1$; obviously, the same relationships hold for
$\kappa^0$ as well. On the other hand, on $\{ \int_0^\infty
|\sigma_t|^2 \ud t < \infty \}$ we have that $\lim_{t \to \infty}
\kappa_t$ exists, and since $2 L_\infty = \int_0^\infty |\sigma_t
(\kappa^1_t - 1/2)|^2 \ud t$ by \eqref{eq: example computation}, we
have $L_\infty < \infty$. We conclude that
\[
\Omega_u  \ = \ \Big \{ \int_0^\infty |\sigma_t|^2 \ud t = \infty
\Big \} \ = \big \{ \liminf_{t \to \infty} \kappa^i_t = 0, \
\limsup_{t \to \infty} \kappa^i_t = 1, \textrm{ for both } i = 0, 1
\big \},
\]
which shows that the result of Proposition \ref{prop: lim cap dist
exists for balanced} cannot be strengthened. It also shows that it
is exactly the unbalanced markets that bring diversity into the
picture and the hope that not all capital will concentrate in one
company only.

\section{The Quasi-Left-Continuous Case} \label{sec: qlc}

We now discuss all the previous results in a more general setting,
where we allow for the processes of company capitalizations to have
jumps. For notions regarding semimartingale theory used in the
sequel, one can consult Jacod and Shiryaev \cite{Jacod - Shiryaev}.
Numbered subsections correspond to previous numbered sections, i.e.,
subsection \ref{subsec: set-up in general case} to section \ref{sec:
ito process model}, subsection 7.2 to section 2, and so on.

\subsection{The set-up} \label{subsec: set-up in general case}
We denote by $S^i$ the capitalization of company $i$. Each $S^i$,
$i=1, \ldots, d$ is modeled as a semimartingale living on an
underlying probability space $\probtriple$, adapted to the
filtration $\filtration$ that satisfies the usual conditions. One
extra ingredient that \emph{has} to be added (in view of Example
\ref{exa: death of company} later) is to allow for the
capitalizations to become zero, which can be considered as
\emph{death}, or \emph{annihilation} of the company. We define the
\textsl{lifetime} of company $i$ as $\zeta^i := \inf \{ t \in
\Real_+ \such S^i_{t-}=0 \textrm{ or } S^i_{t}=0 \}$; each $\zeta^i$
is an $\mathbf{F}$-stopping time. After dying, companies cannot
revive; thus we insist that $S^i_t \equiv 0$, for all $t \geq
\zeta^i$. Note that --- even though individual companies might die
--- we suppose the whole market lives forever; $\max_{1 \leq i \leq d} \zeta^i = + \infty$,
$\prob$-a.s.

We want to write an expression like:
\begin{equation} \label{eq: model, general}
\ud S^i_t = S^i_{t-} \ud X^i_t, \textrm{ for } t \in \dbra{0,
\zeta^i}, \ i = 1, \ldots, d.
\end{equation}
where $\ud X^i_t$ plays the  equivalent r\^ole of $a^i_t \ud t + \ud
M^i_t$ of \eqref{eq: model}. Let us assume for the moment that
$\zeta^i = \infty$ for all $i = 1, \ldots, d$, so that $X^i$ can be
defined as the \textsl{stochastic logarithm} of $S^i$: $X^i :=
\int_0^\cdot (\ud S^i_t / S^i_{t-})$. Then, we know that if we fix
the \textsl{canonical truncation function} $x \mapsto \hx$
($\indic_A$ will denote the indicator of a set $A$), the
\textsl{canonical decomposition} of the semimartingale $X = (X^1,
\ldots, X^d)$ is
\begin{equation} \label{eq: canonical representation}
X \ = \ B + M + [\hx] * (\mu - \eta) + [\hbarx] * \mu.
\end{equation}
In the decomposition \eqref{eq: canonical representation}, $B$ is
predictable and of finite variation; $M$ is a continuous local
martingale; $\mu$ is the \textsl{jump measure} of $X$, i.e., the
random measure on $\Real_+ \times \Real^d$ defined by $\mu ([0,t]
\times A) := \sum_{0 \leq s \leq t} \indic_{A \setminus \{0\}}
(\Delta X_s)$, for $t \in \Real_+$ and $A \subseteq \Real^d$; the
asterisk ``$*$'' denotes integration with respect to random
measures; $\eta$ is the \textsl{predictable compensator} of $\mu$
--- it satisfies $[|x|^2 \wedge 1]*\eta_t < \infty$ for all $t \in \Real_+$, and
$\eta[\Real_+ \times (- \infty, -1)^d] = 0$, since each $S^i$ ($i=1,
\ldots, d$) is constrained to be positive.

Since we do not know a-priori that $\zeta^i = \infty$ for all $i=1,
\ldots, d$, we take the opposite direction of assuming the
representation \eqref{eq: canonical representation}, and pick as
\emph{inputs} a continuous local martingale $M$, a
quasi-left-continuous \textsl{semimartingale jump measure} $\mu$,
and a continuous process $B$ that is locally of finite variation
before a possible explosion to $- \infty$. The continuous local
martingale $M$ being obvious, we remark on the last two objects.

A \textsl{semimartingale jump measure} $\mu$ is a random counting
measure on $\Real_+ \times \Real^d$ with $\mu( \Real_+
 \times \{ 0 \} ) = 0$ and $\mu(\{ t \} \times \Real^d)$
being $\{0, 1\}$-valued for all $t \in \Real_+$, such that its
predictable compensator $\eta$ exists and satisfies $[|x|^2 \wedge
1]*\eta_t < \infty$ for all $t \in \Real_+$. $\mu$ being
quasi-left-continuous means $\mu(\{ \tau \} \times \Real^d) = 0$ for
all \emph{predictable} stopping times $\tau$; this is equivalent to
$\eta(\{ t \} \times \Real^d) = 0$ for all $t \in \Real_+$. In other
words, jumps are permitted as long as they only come in a totally
unpredictable (inaccessible) way. It will also be assumed that
$\mu[\Real_+ \times (- \infty, -1)^d] = 0$ (equivalently,
$\eta[\Real_+ \times (- \infty, -1)^d] = 0$) to keep the
company-capitalization processes positive.

The twist comes for the predictable finite-variation process $B$,
for which we shall assume that its coefficient-processes can explode
to $- \infty$ in finite time. In other words, for each $i = 1,
\ldots, d$ there exists a \emph{strictly increasing} sequence of
stopping times $(\zeta_n^i)_{n \in \Natural}$ such that the stopped
process $\big( B^i_{\zeta_n^i \wedge t} \big)_{t \in \Real_+}$ is
continuous (thus predictable) and of finite variation, and that
$\lim_{n \to \infty} B^i_{\zeta_n^i} = - \infty$. It is clear that
we can choose $\zeta_n^i := \inf \{t \in \Real_+ \ | \ B^i_t = -n
\}$. We further define
\begin{equation} \label{eq: lifetime of company}
\zeta^i \ := \ \big( \lim_{n \to \infty} \uparrow \zeta_n^i \big)
\wedge \inf \{ t \in \Real_+ \ | \ \mu(t,-1) = 1 \}.
\end{equation}
where we write $\mu(t,-1)$ as short for $\mu(\{(t,-1)\})$. The last
definition should be intuitively obvious: annihilation of company
$i$ happens either (a) when $B^i$ explodes to $- \infty$ in which
case we have a continuous transition of $S^i$ to zero in that
$S^i_{\zeta^i -} = 0$, or (b) the first time when $\mu(t,-1) = 1$,
where we have a jump down to zero; in this case we have
$S^i_{\zeta^i -} > 0$ and $S^i_{\zeta^i} = 0$.

Having these ingredients we now \emph{define} the process $X$ via
\eqref{eq: canonical representation}, where we tacitly assume that
$X^i = - \infty$ on $\dbraoo{\zeta^i, \infty}$. We also define the
company capitalizations $S^i$ for $i = 1, \ldots, d$ via \eqref{eq:
model, general}, where we set $S^i = 0$ on $\dbraoo{\zeta^i,
\infty}$.

Setting $C := [M, M]$ to be the \textsl{quadratic covariation}
process of $M$, the triple $(B, C, \eta)$ is called the
\textsl{triplet of predictable characteristics} of $X$. One can find
a continuous, one-dimensional, strictly increasing process $G$ such
that the processes $C$ and $\eta$ are absolutely continuous with
respect to it, in the sense of the equations \eqref{eq: triplet of
pred characteristics} below --- for instance, one can choose $G =
\sum_{i=1}^d [M^i, M^i] + [|x|^2 \wedge 1]*\eta$. We shall also
assume that each $B^i$, $i = 1, \ldots, d$ is absolutely continuous
with respect to $G$ on the stochastic interval $\dbraco{0, \zeta^i}$
--- otherwise it can be shown that there are trivial opportunities
for free lunches of the most egregious kind --- one can check
\cite{KK: num and arbitrage}, Section 5, for more information. It
then follows that we can write
\begin{equation} \label{eq: triplet of pred characteristics}
B = \int_0^\cdot b_t \ud G_t, \ C = \int_0^\cdot c_t \ud G_t,
\textrm{ and } \eta ([0,t] \times E) = \int_0^t \Big( \int_{\Real^d}
\indic_E (x) \nu_u (\ud x) \Big) \ud G_u
\end{equation}
for any Borel subset $E$ of $\Real^d$. Here, all $b$, $c$ and $\nu$
are predictable, $b$ is a vector process, $c$ is a positive-definite
matrix-valued process and $\nu$ is a process with values in the
space of measures on $\Real^d$ that satisfy $\nu(\{0\}) = 0$ and
integrate $x \mapsto 1 \wedge |x|^2$ (so-called \textsl{L\'evy
measures}). Each process $b^i$ for $i = 1, \ldots, d$ is
$G$-integrable on each stochastic interval $\dbra{0 , \zeta_n^i}$,
but on the event $\{ \zeta^i < \infty, \ S_{\zeta^i -} = 0\}$, $b^i$
it is not integrable on $\dbra{0 , \zeta^i}$.

The differential ``$\ud G_t$'' will be playing the r\^ole that
``$\ud t$'' was playing before --- for example, an \textsl{interest
rate process} now is a one-dimensional predictable process $r$ such
that $\int_0^t |r_u| \ud G_u < \infty$ for all $t \in \Real_+$.

\subsection{Perfectly balanced markets} The notion of a perfectly
balanced market is exactly the same as before: we ask that $\kappa$
is a vector $(\prob, \bF)$-martingale.

The first order of business is to find necessary and sufficient
conditions in terms of the triplet $(b, c, \nu)$ for the market to
be perfectly balanced. It\^o's formula gives that the drift part of
the stochastic logarithm process $\int_0^\cdot (\ud \kappa^i_t /
\kappa^i_{t-})$ on $\dbraco{0, \zeta^i}$ is
\[
\int_0^\cdot \bigg(\inner{\e_i - \kappa_{t-}}{b_t - c_t
  \kappa_{t-}}
  + \int_{\Real^d} \bigg[ \frac{\inner{\e_i - \kappa_{t-}}{x}}{1 + \inner{\kappa_{t-}}{x}} - \inner{\e_i - \kappa_{t-}}{x} \indic_{\{ |x| \leq 1 \}} \bigg]
  \nu_t
(\ud x) \bigg)  \ud G_t.
\]
In a perfectly balanced market, this last quantity has to to vanish
--- using same arguments as in the proof of Proposition \ref{prop:
charact of efficiency, cont} we get the following result.

\begin{prop}
The market is perfectly balanced if and only if there exists an
interest rate process $r$ such that the following relationship holds
for each coordinate $i=1, \ldots, d$ on $\dbraco{0, \zeta^i}$:
\begin{equation} \label{eq: drift gen}
b - c \kappa_- + \int \bra{ \frac{x}{1 + \inner{\kappa_-}{x}} - \hx
} \nu (\ud x) = r \fone.
\end{equation}
\end{prop}

Using \eqref{eq: drift gen} above one computes that in a perfectly
balanced market the relative company capitalization $\kappa^i$ for
each $i = 1, \ldots, d$ satisfies
\begin{equation} \label{eq: SDE for perfect bal gen}
\kappa^i = \kappa_0^i \Exp \bigg( \int_0^\cdot \inner{\e_i -
\kappa_{t-}}{\ud M_t} + \bigg[ \frac{\inner{\e_i - \kappa_{-}}{x}}{1
+ \inner{\kappa_{-}}{x}} \bigg] * (\mu - \eta) \bigg), \textrm{ on }
\dbra{0, \zeta^i}
\end{equation}
where $\Exp$ is the stochastic exponential operator.

In order to get a result about existence of perfectly balanced
markets similar to Theorem \ref{thm: existence of perf bal econ,
cont} one has to start with the continuous local martingale $M$ and
a quasi-left-continuous semimartingale jump measure $\mu$ and show
that equations \eqref{eq: SDE for perfect bal gen} have a strong
solution. Below, we show by example that even if we start with an
initial distribution of capital $\kappa_0$ in the \emph{open}
simplex $\Delta^{d-1}$ (so that $\kappa_0^i > 0$ for all $i = 1,
\ldots, d$) and jumps of size $-1$ are not allowed by the jump
measure, annihilation of a company might come at \emph{finite} time
--- stock-killing times were not included just for the sake of
generality, but they come up naturally if possibly unbounded jumps
above are allowed for the company-capitalization processes.

\begin{ex} \label{exa: death of company}
Consider a simple market with two companies (we call them $0$ and
$1$) for which $\kappa^0_0 = \kappa^1_0 = 1/2$, $M \equiv 0$, and
$\mu$ is a jump measure with \emph{at most} one jump at time $\tau$
that is an exponential random variable, and size $l(\tau)$ for a
deterministic function $l$ given by $l(t) = \big( 1 - e^{t/2} / 2
\big)^{-1} \indic_{[0 , 2 \log2)} (t)$. Observe that there is no
jump on $\{ \tau > 2 \log 2\}$, an event of positive probability,
and that $\nu_t (\ud x) = \indic_{(0, \tau]} \delta_{(0, l_t)} (\ud
x)$, where $\delta$ is the Dirac measure.

Now, according to \eqref{eq: SDE for perfect bal gen} the process
$\kappa^1$ should satisfy
\[
\frac{\ud \kappa^1_t}{\ud t} = - \frac{\kappa^1_t (1 - \kappa^1_t)
l_t}{1 + \kappa_t^1 l_t}, \textrm{ for all } t < \tau.
\]
It can be readily checked that the solution of the previous
(ordinary) differential equation for $t < \min \{ \tau, 2 \log 2 \}$
is $\kappa^1 = 1/l$. Thus, on $\{ \tau \geq 2 \log 2 \}$ (which has
positive probability), we have $\kappa^1_t = 0$ for all $t \geq 2
\log 2$, i.e., $\prob[\zeta^1 < \infty] > 0$.
\end{ex}

\begin{thm}
Consider a continuous $(\prob, \bF)$-local martingale $M$ and a
quasi-left-continuous semimartingale jump measure $\mu$. Then, for
any $\F_0$-measurable initial condition $\kappa_0 \equiv
(\kappa_0^i)_{1 \leq i \leq d}$ with $\prob[\kappa_0 \in
\overline{\Delta}^{d-1}] = 1$ the stochastic differential equations
\eqref{eq: SDE for perfect bal gen} have a unique strong solution on
$[0,\infty)$ that lives on $\overline{\Delta}^{d-1}$.

Select any interest-rate process $r$, as well as any $\F_0$-measurable
initial random vector $S_0 = (S_0^i)_{1 \leq i \leq d}$ such that
$S_0^i/ \inner{S_0}{\fone} = \kappa_0^i$. For all $i=1, \ldots, d$,
define $\zeta^i$ by \eqref{eq: lifetime of company} and also define
$b^i$ by \eqref{eq: drift gen} on the interval $\dbraco{0,
\zeta^i}$. With $B = \int_0^\cdot b_t \ud G_t$, if we define $X$ via
\eqref{eq: canonical representation}, then $S$ as defined by
\eqref{eq: model, general} is a model of a perfectly balanced
market.
\end{thm}

\begin{proof}
More or less, one follows the steps of the proof of Theorem
\ref{thm: existence of perf bal econ, cont}, with some twists. We
assume that the initial condition $\kappa_0$ lives on $\Delta^{d-1}$
--- any company $i = 1, \ldots, d$ for which $\kappa^i_0 = 0$ can be safely
disregarded, since then $\kappa^i \equiv 0$.

Set $K_n := [n^{-1}, 1 - n^{-1}]^d$ for all $n \in \Natural$; the
co\"efficients of \eqref{eq: SDE for perfect bal gen} are Lipschitz
on $K_n$. A theorem on strong solutions of stochastic differential
equations involving random measures has to be invoked --- one can
check for example Bichteler \cite{Bichteler: stoch. integration}
(Proposition 5.2.25, page 297) for existence of solutions of
equations of the form \eqref{eq: SDE for perfect bal gen} in the
case of Lipschitz co\"efficients. We infer the existence of an
increasing sequence of stopping times $(\tau_n)_{n \in \Natural}$
such that $\kappa_t \in K_n$ for all $t < \tau_n$ and
$\kappa_{\tau_n} \in \Real^d \setminus K_n$. Using \eqref{eq: SDE
for perfect bal gen} one can show that $\inner{\kappa}{\fone}$ is
constant on $\dbra{0, \tau_n}$ --- since $\kappa_0 \in \Delta^{d-1}$
we have $\kappa_t \in \Delta^{d-1}$ for all $t < \tau_n$. Now, we
claim that $\kappa_{\tau_n} \in \overline{\Delta}^{d-1}$. Since
$\inner{\kappa_{\tau_n}}{\fone} = 1$ we only need show that
$\kappa^i_{\tau_n} \geq 0$ for all $i = 1, \ldots, d$. If $\mu( \{
\tau_n \} \times \Real^d) = 0$ this is trivial. Otherwise, let
$\xi_n \in [-1, \infty)^d$ the (random) point such that $\mu(\tau_n,
\xi_n) = 1$; \eqref{eq: SDE for perfect bal gen} gives
\[
\kappa^i_{\tau_n} = \kappa^i_{\tau_n -} \Big( \frac{1 + \xi_n^i}{1 +
\inner{\kappa_{\tau_n -}}{\xi_n}} \Big) \geq 0,
\]
since $\xi_n^i \geq -1$ and $\inner{\kappa_{\tau_n -}}{\xi_n} > -1$
in view of the fact that $\kappa_{\tau_n -} \in \Delta^{d-1}$.

Pasting solutions together we get that there exists a stopping time
$\tau$ such that $\kappa_t \in \Delta^{d-1}$ for all $t < \tau$ and
$\kappa_{\tau} \in \partial \Delta^{d-1}$ on $\{ \tau < \infty \}$.
Unlike the proof of Theorem \ref{thm: existence of perf bal econ,
cont} we cannot hope now that $\prob [ \tau < \infty ] = 0$, as
Example \ref{exa: death of company} above shows. Rather, we set
$\zeta^i = \tau$ if $\kappa^i_\tau = 0$ for $i = 1, \ldots, d$.

We have constructed a solution to \eqref{eq: SDE for perfect bal
gen} on the stochastic interval $\dbra{0, \tau}$. On the event $\{
\tau < \infty \}$ we continue the construction of the solution to
\eqref{eq: SDE for perfect bal gen} inductively, removing all
companies that have died. In at most $d - 1$ steps we either have
constructed the solution for all $t \in \Real^d$, or only one
company (say, $i$) has remained in which case we shall have $\kappa
= \e_i$ from then onwards.
\end{proof}

\subsection{Perfect balance and growth}

Growth-optimality of a portfolio and the market are now defined, and
their relation to perfect balance is established.

A portfolio is a $d$-dimensional predictable and $X$-integrable
processes, and from now onwards we restrict attention to the
$\C$-constrained class $\Pi_\C$ where $\C \equiv
\overline{\Delta}^{d-1}$. If $V^\pi$ denotes the wealth process
generated by $\pi$ we have
\[
\frac{\ud V^\pi_t}{V^\pi_{t-}} \ = \ \sum_{i=1}^d \pi_t^i \frac{\ud
S^i_t}{S^i_{t-}} + \Big( 1 - \sum_{i=1}^d \pi_t^i \Big) r_t \ud G_t,
\]
where $r$ is some interest rate process coming from a bank in the
market.

The market portfolio is not $\kappa$ now, but rather its
left-continuous version $\kappa_-$ (the vector process $\kappa$ as
appears in (\ref{eq: market portf cont}) is not in general
predictable, but only adapted and right-continuous). It is trivial
to check that $V^{\kappa_-} = \inner{S}{\fone}/\inner{S_0}{\fone}$.

The concept of growth of a portfolio is sometimes not well-defined,
as the log-wealth process $\log V^\pi$ might not be a special
semimartingale, which means that its finite-variation part fails to
exist. In order to define a growth optimal portfolio $\rho$, we use
the idea contained in Remark \ref{rem: on numeraire property of
growth optimal}: we ask that $V^\pi / V^\rho$ is a supermartingale
for all $\pi \in \Pi_\C$. It turns out (one can check \cite{KK: num
and arbitrage}, for example) that this requirement is equivalent to
$\rel(\pi | \rho) \leq 0$ for all $\pi \in \Pi_\C$, where the
\textsl{relative rate of return} process is
\begin{equation} \label{eq: rel rate of ret}
\rel(\pi | \rho) \ := \ \inner{\pi - \rho}{b - r \fone} - \inner{\pi
- \rho}{c \rho} + \int \Big[ \frac{\inner{\pi - \rho}{x}}{1 +
\inner{\rho}{x}} - \inner{\pi - \rho}{x} \indic_{\{ |x| \leq 1 \}}
\Big] \nu (\ud x).
\end{equation}
The market will be called a growth market if $\kappa_-$ is
growth-optimal according to this last definition. It is easily shown
that in order to have a growth market we must have \eqref{eq: drift
gen} holding, where $r$ is now the banking interest rate.

Exactly the same remarks on interest rates hold as the ones in
subsection \ref{subsec: interest rate} --- the concepts of perfect
balance and growth in markets are thus equivalent.

\subsection{Balanced markets}

To define the loss-of-balance process, let $\rho$ be the
growth-optimal portfolio in the class $\Pi_\C$ and set
\[
L := \int_0^\cdot \big(- \rel(\kappa_{t-} | \rho_t) + \frac{1}{2}
c_t^{\kappa_- | \rho} \big) \ud G_t + \Big[ 1 \wedge \Big| \log
\frac{1+\inner{\kappa_-}{x}}{1+\inner{\rho}{x}} \Big|^2 \Big]
* \eta,
\]
where we define $c^{\pi_1 | \pi_2} := \inner{\pi_2 - \pi_1}{c(\pi_2 - \pi_1)}$
for two portfolios $\pi_1$ and $\pi_2$. As before, set $\Omega_b :=
\{ L_\infty < \infty\}$ and $\Omega_u := \Omega \setminus \Omega_b =
\{ L_\infty = \infty\}$. The above definition of $L$ is slightly
different than the one of Definition \ref{dfn: balanced econ} for
the case of It\^o processes, but for this special case it is easy to
see that the sets $\Omega_b$ and $\Omega_u$ that are obtained using
the two definitions are the same --- and this is the only thing of
importance.

With a little help from Lemma \ref{lem: pos mart conv gen} (more
precisely, its generalization discussed in Remark
\ref{rem: pos supermart conv}) we get that $\Omega_b = \{ \lim_{t
\to \infty} (V_t^{\kappa_-} / V^\rho_t)
> 0 \}$ and $\Omega_u = \{ \lim_{t \to \infty} (V_t^{\kappa_-} /
V^\rho_t) = 0 \}$. Based on this characterization of the event of
balanced outcomes, Theorem \ref{thm: charact of weak eff} can be
proved for our more general case now.

\subsection{Limiting capital distribution of balanced markets}

Of course, the event-inclusion $\Omega_b \ \subseteq \ \{
\kappa_\infty := \lim_{t \to \infty} \kappa_t \textrm{ exists} \}$
follows exactly from the equivalent of Theorem \ref{thm: charact of
weak eff} in the quasi-left-continuous case --- a limiting capital
distribution exists for the balanced outcomes.

Two companies are equivalent (we write $i \sim_\omega j$) if $d^{i |
j} (\omega) = \infty$, where
\begin{equation} \label{eq: dist of companys gen}
d^{i|j} := \int_0^\infty \big( |\rel(\e_i | \rho_t) - \rel(\e_j |
\rho_t)| + \frac{1}{2} c^{i | j}_t \big) \ud G_t + \Big[ 1 \wedge
\Big| \log \frac{1+ x^i}{1+ x^j} \Big|^2 \Big]
* \eta_\infty
\end{equation}
is a measure of distance between two companies in an
$\omega$-by-$\omega$ basis. Again, this definition does not fully
agree with the one given in \eqref{eq: dist of companys}, but it is
easy to see that the events $\{ i \sim j \}$ are identical under
both definitions. Segregated markets and the segregation set
$\Sigma$ are formulated exactly as in Definition \ref{dfn: dist of
sect con}.

We again have $\Omega_b \cap \{ \kappa^i_\infty
> 0, \ \kappa^j_\infty > 0 \} \subseteq \{ i \sim j \}$. The proof
follows the steps of Lemma \ref{lem: key for divers failure},
invoking Lemma \ref{lem: pos mart conv gen} (actually, Remark
\ref{rem: pos supermart conv}) from the Appendix. Then, Theorem
\ref{thm: one company takes all} follows trivially: on balanced
outcomes that segregation of companies holds, one company will take
all.

\appendix

\section{Limiting Behavior of Local Martingales}

The proof of the following result is well-known for continuous-path
semimartingales --- for the slightly more general case described
below, the proof is the same.

\begin{lem} \label{prop: strong law of large numbers cont mart}
Let $X = M + x*(\mu - \eta)$ be a local martingale, where $M$ is a
continuous local martingale and $\mu$ is the jump measure of $X$
with $\eta$ its predictable compensator. We assume that $X$ has
bounded jumps: $|\Delta X| \leq c$ for some constant $c \geq 0$.
Then, with $B := [M,M] + |x|^2*\eta$ we have $\{ \lim_{t \to \infty}
X_t \textrm{ exists in } \Real \} = \{ B_\infty < + \infty \}$,
while on the event $\{ B_\infty = + \infty \}$ we have $\lim_{t \to
\infty} (X_t / B_t) = 0$.
\end{lem}

%For the continuous case,
This allows one to prove the following lemma.

\begin{lem} \label{lem: asymptotic pos loc mart, cont}
For a continuous local martingale  $M$, consider the exponential
local martingale $\Exp(M) = \exp(M - [M,M]/2)$. Then,
$\Exp(M)_\infty := \lim_{t \to \infty} \Exp(M)_t$ exists and is
$\Real_+$-valued. Further, $\{ [M,M]_\infty < + \infty \} = \{
\Exp(M)_\infty > 0 \}$.
\end{lem}

\begin{proof}
Existence of $\Exp(M)_\infty$ follows from the supermartingale
convergence theorem. Lemma \ref{prop: strong law of large numbers
cont mart} gives $\{ [M,M]_\infty < + \infty \} = \{ \lim_{t \to
\infty} \Exp(M)_t \in \Real \}$; thus $\{ [M,M]_\infty < + \infty \}
\subseteq \{ \Exp(M)_\infty
> 0 \}$. For the other inclusion, Lemma \ref{prop: strong law of
large numbers cont mart} again gives that on $\{ [M,M]_\infty = +
\infty \}$ we have $\lim_{t \to \infty} (\log \Exp(M)_t / [M,M]_t) =
- 1/2$; this means that $\lim_{t \to \infty} \log \Exp(M)_t = -
\infty$, or $\Exp(M)_\infty = 0$ and we are done.
\end{proof}

In order to prove the equivalent of Lemma \ref{lem: asymptotic pos
loc mart, cont} for general semimartingales, a ``strong law of large
numbers'' result for increasing processes will be needed.

\begin{lem} \label{lem: slln for increasing}
Let $A$ be an increasing, right-continuous and adapted process with
$|\Delta A| \leq c$ for some constant $c > 0$, and let $\tilde{A}$
be its predictable compensator, so that $A - \tilde{A}$ is a local
martingale. Then, we have $\{ A_\infty < \infty \} =
\{\tilde{A}_\infty < \infty\}$ and on $\{\tilde{A}_\infty = \infty
\}$ we have $\lim_{t \to \infty} ( A_t / \tilde{A}_t ) = 1$.
\end{lem}

\begin{proof}
It is easy to see that we can assume without loss of generality that
$A$ is pure-jump and quasi-left-continuous (if not, decompose $A$
into a part as described and another part that is predictable; this
second part can be subtracted from both $A$ and $\tilde{A}$). Let
$\eta$ be the predictable compensator of the jump measure of $A$;
observe then that $\tilde{A} = x * \eta$ and if $N := A -
\tilde{A}$, then $B := \widetilde{[N,N]} = \abs{x}^2 * \eta$. Since
$A$ has jumps bounded by $c$, it is clear that $B \leq c \tilde{A}$.

On $\{\tilde{A}_\infty < + \infty \}$ we have $B_\infty < + \infty$,
so that $N_\infty$ exists, and thus $A_\infty < + \infty$.

Now, work on $\{\tilde{A}_\infty = + \infty \}$. If $B_\infty < +
\infty$, $M_\infty$ exists and is real-valued, so obviously $\lim_{t
\to \infty} ( A_t - \tilde{A}_t )/\tilde{A}_t = 0$. If $B_\infty = +
\infty$, we have $\lim_{t \to \infty} ( A_t - \tilde{A}_t ) / B_t =
0$, so that also $\lim_{t \to \infty} ( A_t - \tilde{A}_t
)/\tilde{A}_t = 0$, and this completes the proof.
\end{proof}

\begin{lem} \label{lem: pos mart conv gen}
Let $X$ and $Y$ be local martingales with $\Delta X > -1$, $\Delta Y
> -1$ (then, $\Exp(X)$ and $\Exp(Y)$ are positive local
martingales). Write $X = M + x*(\mu - \eta)$ and $Y = N + y*(\mu -
\eta)$ with $M$ and $N$ being continuous local martingales, $\mu$
the 2-dimensional jump measure of $(X,Y)$ and $\eta$ its predictable
compensator. Then,
\begin{enumerate}
    \item $\{ \Exp(X)_\infty > 0 \} \ = \ \{ [M, M]_\infty/2
    + [1 \wedge \abs{\log(1+x)}^2] * \eta_\infty < + \infty \}$.
    \item $\{ \Exp(X)_\infty > 0, \ \Exp(Y)_\infty > 0 \} \ \subseteq \ \{ d^{X | Y} < +
    \infty \}$,  where we have set
    \[
    d^{X|Y} := \frac{1}{2}[M - N, M - N]_\infty
    + \Big[ 1 \wedge \Big| \log \Big( \frac{1+x}{1+y} \Big) \Big|^2 \Big] * \eta_\infty
    \]
\end{enumerate}
\end{lem}

\begin{proof}
For (1), the definition of the stochastic exponential gives
\[
\log \Exp(X) = X - \frac{1}{2} [M, M] - [x - \log (1+x)] * \mu.
\]
Since $\Delta \log \Exp(X) = \log(1 + \Delta X)$, on $\{
\Exp(X)_\infty > 0 \} = \{\log \Exp(X)_\infty \in \Real \}$ we
should have $| \log(1 + \Delta X_t) | > 1$ for a finite
(path-dependent) number of $t \in \Real_+$ --- equivalently, we must
have that $\indic_{\{ |\log(1+x)| > 1\}} * \mu_\infty < +\infty$ and
then Lemma \ref{lem: slln for increasing} implies $\indic_{\{
|\log(1+x)|
> 1\}} * \eta_\infty < +\infty$. Now, if we subtract the
semimartingale $[\log(1+x) \indic_{\{ |\log(1+x)| > 1\}} ] *
\mu_\infty$ (which is actually only a finite sum) from $\log
\Exp(X)$, what remains is a semimartingale with bounded (by one)
jumps. The canonical representation of the semimartingale $\log
\Exp(X) - [\log(1+x) \indic_{\{ |\log(1+x)| > 1\}} ] * \mu$ into a
sum of a predictable finite-variation part (first two terms in
\eqref{eq: decomp of semimart minus jumps} below) and a local
martingale part (last two terms):
\begin{equation} \label{eq: decomp of semimart minus jumps}
- \frac{1}{2} [M, M] + [x - \log(1+x) \indic_{\{\abs{\log(1+x)} \leq
1\}}] * \eta + M + [\log(1+x) \indic_{\{\abs{\log(1+x)} \leq 1\}}] *
(\mu - \eta).
\end{equation}
This last semimartingale must have a real limit at infinity. Observe
that on $\{ [M,M]_\infty + [|\log(1+x)|^2 \indic_{\{ |\log(1+x)|
\leq 1 \}}] * \eta_\infty = + \infty \}$ this cannot happen, because
Lemma \ref{prop: strong law of large numbers cont mart} would give
that the limit at infinity of the ratio of (\ref{eq: decomp of
semimart minus jumps}) to its predictable finite variation part
would be equal to 1, which would imply that the semimartingale
(\ref{eq: decomp of semimart minus jumps}) does not have a limit.
This completes the proof of (1).

Let us proceed to (2); we work on $\{ \Exp(X)_\infty > 0, \
\Exp(Y)_\infty > 0 \}$. Part (1) of this lemma gives $[M-N,
M-N]_\infty \leq 2 [M,M]_\infty + 2 [N,N]_\infty < + \infty$. Now,
define
\[
\Lambda \ := \ \Big \{ (x,y) \in (-1, \infty)^2 \ | \ \Big| \log
\Big( \frac{1+x}{1+y} \Big) \Big| \leq 1 \Big \}
\]
as well as $\Lambda_x := \{(x,y) \in (-1, \infty)^2 \ | \ |\log
(1+x)| \leq 1/2 \}$ and $\Lambda_y := \{(x,y) \in (-1, \infty)^2 \ |
\ |\log (1+y)| \leq 1/2 \}$. With the prime ``$'$'' denoting the
complement of a set, we have $\Lambda' \subseteq \Lambda'_x \cup
\Lambda'_y$, so $\indic_{\Lambda'} * \eta_\infty < + \infty$ as
discussed before. We then only have to show that $[\indic_\Lambda
|\log ((1+x)/(1+y))|^2]
* \eta_\infty < + \infty$. Since we have that $[\indic_{\Lambda_x} |\log (1+x)|^2]
* \eta_\infty < + \infty$ and $[\indic_{\Lambda_y} |\log (1+y)|^2]
* \eta_\infty < + \infty$ holds from part (1) of this lemma,
we need only show that $|f|^2 * \eta_\infty < + \infty$, where
\[
f(x,y) := \indic_\Lambda \log \Big( \frac{1+x}{1+y} \Big) -
\indic_{\Lambda_x} \log (1+x) + \indic_{\Lambda_y} \log (1+y).
\]
It is clear from part (1) that $[\indic_{\Lambda'} |f|^2]*
\eta_\infty < + \infty$. Now, on $\Lambda \cap \Lambda_x \cap
\Lambda_y$ we have $f = 0$, while on $\Lambda \cap \Lambda'_x \cap
\Lambda'_y$ we have $|f| \leq 1$. For $(x, y) \in \Lambda \cap
\Lambda_x \cap \Lambda'_y$ we have $f(x,y) = - \log(1+y)$, which
(using the triangle inequality) cannot be more than $3/2$ in
absolute value. The similar thing holds on $\Lambda \cap \Lambda'_x
\cap \Lambda_y$, so finally $[\indic_\Lambda |f|^2]*\eta_\infty \leq
(3/2) [\indic_{\Lambda \cap (\Lambda_x \cap
\Lambda_y)'}]*\eta_\infty < + \infty$, which completes the proof.
\end{proof}

\begin{rem} \label{rem: pos supermart conv}
Lemma \ref{lem: pos mart conv gen} can be extended in the case where
$X$ and $Y$ are of the form $X = -A + M + x*(\mu - \eta)$ and $Y = -
B + N + y*(\mu - \eta)$, where $A$ and $B$ are \emph{increasing} and
\emph{continuous} adapted processes. In that case we have
\begin{enumerate}
    \item $\{ \Exp(X)_\infty > 0 \} \ = \ \{ A_\infty + [M,
    M]_\infty/2 + [1 \wedge \abs{\log(1+x)}^2] * \eta_\infty < + \infty \}$.
    \item $\{ \Exp(X)_\infty > 0, \ \Exp(Y)_\infty > 0 \} \ \subseteq \ \{ d^{X | Y} < +
    \infty \}$, where
    \[
    d^{X|Y} := \int_0^\infty \ud |A - B|_t + \frac{1}{2}[M - N, M - N]_\infty
    + \Big[ 1 \wedge \Big| \log \Big( \frac{1+x}{1+y} \Big) \Big|^2 \Big] * \eta_\infty
    \]
\end{enumerate}
We can extend the discussion further when $A$ or $B$ might explode
to $\infty$ in finite time, i.e., if the lifetimes $\zeta^X := \inf
\{ t \in \Real_+ | X_t = - \infty \}$ and $\zeta^Y := \inf \{ t \in
\Real_+ | Y_t = - \infty \}$ are finite, exactly as described in
subsection \ref{subsec: set-up in general case} of the main text.
\end{rem}

% ----------------------------------------------------------------
\bibliographystyle{amsplain}

\end{document}